\nonstopmode
\documentclass[11pt]{article}

\usepackage{amsmath,amssymb,amsthm}
\usepackage{graphics,epsfig,calc}
\usepackage{amsmath}
\usepackage{latexsym,epsfig,bm,amssymb}
\usepackage{color}
\usepackage{amsthm,mathrsfs}

\usepackage{mathptmx}

\DeclareSymbolFont{AMSb}{U}{msb}{m}{n}
\DeclareSymbolFontAlphabet{\mathbb}{AMSb}

\newcommand{\beqn}{\begin{eqnarray}}
\newcommand{\eeqn}{\end{eqnarray}}
\newcommand{\be}{\begin{equation}}
\newcommand{\ee}{\end{equation}}
\newcommand{\ba}{\begin{array}}
\newcommand{\ea}{\end{array}}
\newcommand{\bo}{{\hfill\loota}}
\newcommand{\loota}{\hbox{\enspace{\vrule height 7pt depth 0pt width 7pt}}}

\newcommand{\cD}{{\cal D}}

\newcommand{\cE}{{\cal E}}
\newcommand{\cF}{{\cal F}}
\newcommand{\cH}{{\cal H}}

\newcommand{\cM}{{\cal M}}
\newcommand{\cO}{{\cal O}}

\newcommand{\cS}{{\cal S}}

\newcommand{\cV}{{\cal V}}
\newcommand{\cW}{{\cal W}}

\newcommand{\cX}{{\cal X}}

\newcommand{\al}{\alpha}

\newcommand{\ci}{\cite}
\newcommand{\de}{\delta}
\newcommand{\De}{\Delta}
\newcommand{\ds}{\displaystyle}
\newcommand{\fr}{\frac}
\newcommand{\ga}{\gamma}

\newcommand{\la}{\label}
\newcommand{\lam}{\lambda}
\newcommand{\Lam}{\Lambda}
\newcommand{\na}{\nabla}

\newcommand{\om}{\omega}
\newcommand{\vp}{\varphi}

\newcommand{\ov}{\overline}

\newcommand{\pa}{\partial}
\newcommand{\re}{\ref}

\newcommand{\Si}{\Sigma}
\newcommand{\si}{\sigma}
\newcommand{\ti}{\tilde}
\newcommand{\dist}{\rm dist\5}
\newcommand{\ve}{\varepsilon}

\newcommand\C{{\mathbb C}}
\newcommand\R{{\mathbb R}}
\newcommand\N{{\mathbb N}}
\newcommand\Z{{\mathbb Z}}

\newcommand{\Ga}{\Gamma}

\newcommand{\vka}{\varkappa}

\newcommand{\cm}{{\rm m}}
\newcommand\nab{{\bf \nabla}}

\newcommand{\5}{{\hspace{0.5mm}}}

\newcommand{\3}{{\hspace{0.2mm}}}

\newcommand{\const}{\mathop{\rm const}\nolimits}

\newcommand{\rRe}{{\rm Re\5}}
\newcommand{\rIm}{{\rm Im\5}}

\renewcommand{\Pr}{\hspace{-6mm}{\bf Proof~}}

\newcommand{\Ran}{{\rm Ran\3}}

\renewcommand{\theequation}{\thesection.\arabic{equation}}
\newtheorem{theorem}{Theorem}[section]
\renewcommand{\thetheorem}{\arabic{section}.\arabic{theorem}}
\newtheorem{definition}[theorem]{Definition}

\newtheorem{lemma}[theorem]{Lemma}
\newtheorem{example}[theorem]{Example}
\newtheorem{remark}[theorem]{Remark}
\newtheorem{remarks}[theorem]{Remarks}
\newtheorem{cor}[theorem]{Corollary}
\newtheorem{proposition}[theorem]{Proposition}

\newcommand{\bd}{\begin{definition}}
\newcommand{\ed}{\end{definition}}
\newcommand{\bt}{\begin{theorem}}
\newcommand{\et}{\end{theorem}}
\newcommand{\bqt}{\begin{qtheorem}}
\newcommand{\eqt}{\end{qtheorem}}

\newcommand{\bp}{\begin{proposition}}
\newcommand{\ep}{\end{proposition}}

\newcommand{\bl}{\begin{lemma}}
\newcommand{\el}{\end{lemma}}
\newcommand{\bc}{\begin{cor}}
\newcommand{\ec}{\end{cor}}

\newcommand{\bex}{\begin{example}}
\newcommand{\eex}{\end{example}}
\newcommand{\bexs}{\begin{examples}}
\newcommand{\eexs}{\end{examples}}

\newcommand{\bexe}{\begin{exercice}}
\newcommand{\eexe}{\end{exercice}}

\newcommand{\br}{\begin{remark} }
\newcommand{\er}{\end{remark}}
\newcommand{\brs}{\begin{remarks}}
\newcommand{\ers}{\end{remarks}}
\begin{document}

\begin{titlepage}

\begin{center}
{\Large\bf On linear stability of crystals in 
\medskip\\
 the Schr\"odinger--Poisson model
}
\end{center}
\bigskip\bigskip

\begin{center}
{\large A. Komech}
\footnote{
Supported partly by 
Austrian Science Fund (FWF): P28152-N35,
and the grant of  RFBR 13-01-00073.}
\\
{\it Faculty of Mathematics of Vienna University\\
and Institute for Information Transmission Problems RAS } \\
e-mail:~alexander.komech@univie.ac.at
\bigskip\\
{\large E. Kopylova}
\footnote{
Supported partly by  
Austrian Science Fund (FWF): P27492-N25,
and the grant of  RFBR 13-01-00073.}
\\
{\it Faculty of Mathematics of Vienna University\\
and Institute for Information Transmission Problems RAS} \\
e-mail:~elena.kopylova@univie.ac.at
\end{center}
\vspace{1cm}

\begin{abstract}

\end{abstract} We consider the 
Schr\"odinger--Poisson--Newton equations for crystals
with a cubic lattice and one ion per cell.
We linearize this dynamics at the ground state and
introduce a novel class of the ion charge densities which provide
the stability of the linearized dynamics. 
This is the first result on linear stability for crystals. 

Our key result
is the {\it energy positivity}
for the 
Bloch generators 
of the 
linearized dynamics  under a
Wiener-type condition 
on the ion charge density.
We also assume  an additional condition which cancels
the negative contribution caused by electrostatic instability.

The proof of the energy positivity relies on a 
novel factorization of the corresponding Hamilton functional.
We show that the  energy positivity can fail if 
the additional condition  breaks down
while the Wiener condition holds.

The Bloch generators  are 
nonselfadjoint (and even nonsymmetric) Hamilton operators. 
We diagonalize these generators using our theory of  
spectral resolution of the Hamilton ope\-rators {\it with positive definite energy} 
\ci{KK2014a,KK2014b}.
Using this  spectral resolution, we establish the stability
of the linearized crystal dynamics.



\end{titlepage}

%%%%%%%%%%%%%%%%%%%%%%%%%%%%%%%%%%%%%%%%%%%%%%%
\section{Introduction}
%%%%%%%%%%%%%%%%%%%%%%%%%%%%%%%%%%%%%%%%%%%%%%%
First mathematical results on stability of matter were obtained 
by Dyson and Lenard 
in \ci{D1967, DL1968} 
where the energy bound from below
has been established.
The thermodynamic limit
for the Coulomb systems  
was studied first by Lebowitz and Lieb 
\ci{LL1969,LL1973}, see the survey and further development in \ci{LS2010}.
These results were extended  by Catto, L. Lions, Le Bris 
and others to Thomas-Fermie and Hartree-Fock models
\ci{CBL1998,CBL2001,CBL2002}. All these results concern  either
the convergence of
the ground state of finite particle systems
in the thermodynamic limit or 
the existence of the ground state
for infinite particle systems. 
The dynamical stability of infinite particle ground states
was never considered previously. 

We establish 
for the first time
the dynamical stability of crystal ground state
in linear approximation for the simplest Schr\"o\-din\-ger-Poisson model.
The ground state for this model was constructed in \ci{K2014}.

We consider crystals with the cubic
lattice  $\Ga= \Z^3$ and with one ion per cell.
 The electron cloud is described by  one-particle 
Schr\"odinger equation.
The ions are described as classical particles   
that corresponds to the  
Born and Oppenheimer  approximation.
The ions interact with the electron cloud via 
the scalar potential, which  is a solution to the corresponding Poisson equation.

This model
does not respect the Pauli exclusion principle for electrons.
However, it provides a convenient framework 
to introduce suitable functional tools, which might be useful 
for physically 
more realistic models (Thomas--Fermie, Hartree--Fock, and second quantized models). 
In particular, we find a novel Wiener-type 
stability criterion (\re{W1}), (\re{Wai}).

This investigation is motivated by the lack of a suitable mathematical model for a
rigorous analysis 
of fundamental quantum phenomena in the solid state physics: 
heat conductivity, electric conductivity, thermoelectronic emission, photoelectric effect, 
Compton effect, 
etc., see \ci{BLR}.
\medskip

%%%%%%%%%%%%%%%%%%%%%%%%%%%%%%%%%%%%%%%%%%%%%%%%%%%%%%%%%%%%%%%%%%%%%%%%%%%%%%%%%%
We denote by  
$\sigma(x)$ the charge density of one 
ion:
\be\la{ro+}
\int_{\R^3} \sigma(x)dx=eZ>0, 
\ee 
where $e>0$ is the elementary charge.
Let $\psi(x,t)$ be the wave function of the electron field, and
$\Phi(x)$ be the electrostatic  potential generated by the ions and electrons.
We assume $\hbar=c=\cm=1$, where $c$ is the speed of light and $\cm$ is the electron mass.
Then the coupled equations  read
\beqn\la{LPS1}
i\dot\psi(x,t)\!\!&=&\!\!-\fr12\De\psi(x,t)-e\Phi(x,t)\psi(x,t),\qquad x\in\R^3,
\\
\nonumber\\
-\De\Phi(x,t)\!\!&=&\!\!\rho(x,t):=\sum_n\sigma(x-n-q(n,t))-e|\psi(x,t)|^2,\qquad x\in\R^3,
\la{LPS2}
\\
\nonumber\\
M\ddot q(n,t)
\!\!&=&\!\!-\langle\nab\Phi(x,t),\sigma(x-n-q(n,t))\rangle, 
\qquad n\in\Z^3.
\la{LPS3}
\eeqn
Here the 
brackets
 stand for the Hermitian scalar product in the Hilbert
space $L^2(\R^3)$ and for its different extensions, and the series (\re{LPS2}) converges in 
a suitable sense.
All derivatives here and below are understood in the sense of distributions.
These equations can be  written as the Hamilton system
with a formal
Hamilton  functional
\be\la{Hfor}
 \cH(\psi,q,p)=\fr12\int_{\R^3}[|\na\psi(x)|^2+\rho(x)G\rho(x)]dx+\sum_{n} \fr{p^2(n)}{2M},
\ee
where $G:=-\De^{-1}$ and
$ q:=(q(n): ~n\in\Z^3)$, $p:=(p(n): ~n\in\Z^3)$, and $\rho(x)$ is defined 
similarly to
(\re{LPS2}). 
Namely, the system (\re{LPS1})-(\re{LPS3}) can 
be formally written as 
\be\la{HSi}
i\dot \psi(x,t)=\pa_{\ov \psi(x)}\cH,
~~~
\dot q(n,t)=\pa_{p(n)}\cH,
~~~
\dot p(n,t)=-\!\pa_{q(n)}\cH,
\ee
where $\pa_{\ov z}:=\fr12(\pa_{z_1}+i\pa_{z_2})$ with $z_1=\rRe z$ and $z_2=\rIm z$.
A ground state of a crystal is a $\Ga$-periodic stationary solution  
\be\la{gr}
\psi^0(x)e^{-i\om^0 t}~,~~~ \Phi^0(x)~,~~~~q^0(n)=q^0~~{\rm for}~~ n\in\Z^3
\ee
with a real $\om^0$
(and $q^0\in\R^3$ can be chosen arbitrary).
A ground state was  constructed in \ci{K2014}.
Substituting (\re{gr}) into (\re{LPS1})-(\re{LPS3}), we obtain the system
\beqn\la{LPS10}
\om^0\psi^0(x)\!\!&=&\!\!-\fr12\De\psi^0(x)-e\Phi^0(x)\psi^0(x),\qquad x\in T^3:=\R^3/\Ga,
\\
\nonumber\\
-\De\Phi^0(x)\!\!&=&\!\!\rho^0(x):=\si^0(x)-e|\psi^0(x)|^2,\qquad x\in T^3,
\la{LPS20}
\\
\nonumber\\
\la{LPS30}
0\!\!&=&\!\!- \5\langle\nab\Phi^0(x),
\si(x-n-q^0)\rangle,\qquad n\in\Z^3,
\eeqn
where we denote
\be\la{ro+2}
\si^0(x):=
\sum_{n}\si(x-n-q^0).
\ee
In present paper we prove
the 
stability
for the 
{\it formal linearization}
of the nonlinear  system (\re{LPS1})-(\re{LPS3})
at the ground state (\re{gr}).
Namely, 
substituting
\be\la{lin2i}
 \psi(x,t)=[\psi^0(x)+\Psi(x,t)]e^{-i\om^0 t},~~~~ q(n,t)=q^0+Q(n,t)
\ee
into the nonlinear equations (\re{LPS1}), (\re{LPS3}) 
with $\Phi(x,t)=G\rho(x,t)$,
we {\it formally} obtain  the linearized equations (see Appendix A)
\be\la{LPS1Li}
\!\!\!\!\!\!\!\!\!\!
\left.\ba{rcl}
\!\!\!\! &\!\!\!\!&\!\!\!\![i\pa_t+\om^0]\Psi(x,t)=-\fr12\De\Psi(x,t)-e\Phi^0(x)\Psi(x,t)-e\psi^0(x)G\rho_1(x,t)\\\\
\!\!\!\! &\!\!\!\!&\!\!\!\!\dot Q(n,t)=P(n,t)/M\\\\
\!\!\!\!&\!\!\!\!&\!\!\!\! \dot P(n,t)=-\langle\nab G\rho_1(t), \sigma(x-n-q^0)\rangle
+\langle\na\Phi^0,\na\si(x-n-q^0)   Q(n,t)\rangle
\ea\right|
\ba{c}x\in\R^3\\
n\in\Z^3
\ea
\ee
Here  
$\rho_1(x,t)$ is the linearized charge density
\be
 \rho_1(x,t)=-\sum_n\na
 \si(x-n-q^0) Q(n,t)-2e\rRe[\psi^0(x)\ov{\Psi(x,t)}],\la{ro1i}
\ee
The system (\re{LPS1Li}) is linear over $\R$ but it is not complex linear. 
This is due to the  last term in (\re{ro1i}),  which appears from 
the linearization of the term  $|\psi|^2=\psi\ov\psi$ in (\re{LPS2}).
However, we need the complex linearity for the application of the spectral theory.
This why we will consider below the complexification  of the system (\re{LPS1Li})
writing it in the variables 
$\Psi_1(x,t):=\rRe\Psi(x,t),\Psi_2(x,t):=\rIm \Psi(x,t)$.
We will consider the case when the ground state $\psi^0(x)$ can be 
taken to be a
 real function.
In this case
\be\la{realdot}
\rRe[\psi^0(x)\ov{\Psi(x,t)}]=\psi^0(x)\Psi_1(x,t).
\ee
Further we denote  
\be\la{Yi}
Y(t)=(\Psi_1(\cdot,t),\Psi_2(\cdot,t),Q(\cdot,t),P(\cdot,t)).
\ee
Then (\ref{LPS1Li})  can be written as
\be\la{JDi}
\dot Y(t)=AY(t),\qquad
 A=\left(\ba{ccrl}
 0   &  H^0 &  0  & 0\medskip\\
-H^0-2e^2\psi^0G\psi^0 &   0  & -S  & 0\\
        0                         &           0                        &   0    &   M^{-1}\\
     -2S^{\5*}                   &              0            &  -T    &  0\\
\ea\right),
\ee
where  
$H^0:=-\fr12\De-e\Phi^0(x)-\om^0$,
the operators $S$ and $T$ correspond to 
matrices 
(\re{S}) and (\re{T}) respectively, and
$\psi^0$ 
denotes the operators of  multiplication 
by the real function $\psi^0(x)$.
The Hamilton representation (\re{HSi}) implies that 
\be\la{AJB}
A=JB,\qquad B=D^2\cH(\psi^0,q^0,0)=\left(\ba{cccl}
 2H^0+4e^2\psi^0 G\psi^0 & 0 & 2S & 0
 \medskip\\
 0 & 2H^0  &0 & 0\medskip\\
 2S^{\5*}  &    0  &   T    & 0  \\
      0      &    0            &   0    &  M^{-1} \\
\ea\right),
\ee
 where $J$ is  the skew-symmetric matrix  (\ref{J}).
 Our basic result is the stability 
for the linearized system (\ref{JDi}):
for any finite energy initial state 
there exists a unique  global solution, and it is bounded in the energy norm.
\medskip

We show that the generator $A$ is densely defined 
in the Hilbert space $\cX:=L^2(\R^3)\oplus L^2(\R^3)\oplus\R^3\oplus\R^3$
and
commutes with  translations by vectors from $\Ga$.
Hence, the  equation 
(\re{JDi}) can be reduced by the Fourier--Bloch--Gelfand--Zak
transform 
to  equations
with the corresponding Bloch generators  $\ti A(\theta)=J\ti B(\theta)$,
which depend on the parameter $\theta$ from the Brillouin zone
 $\Pi^*:=[0,2\pi]^3$. The Bloch energy operator $\ti B(\theta)$ is given by  
\be\la{hess2i}
 \ti B(\theta) =\!\left(\!\ba{cccl}
 2\ti H^0(\theta)+4e^2\psi^0 \ti G(\theta)\psi^0&  0 & 2\ti S(\theta)&0
 \medskip\\
 0 &2\ti H^0(\theta)&  0 &0
\medskip\\
 2\ti S^{\5*}(\theta) &   0
 &   \hat T(\theta)             & 0  \\
      0    &    0            &   0 &  M^{-1} \\
\ea\!\right),\qquad 
\theta\in\Pi^*
\setminus\Ga^*,
\ee
 where $\Ga^*:=2\pi\Z^3$, and
$\ti H^0(\theta):=-\fr12(\na+i\theta)^2-e\Phi^0(x)-\om^0$. 
Further, $\ti G(\theta)$ is the inverse to the 
operator 
$(i\na-\theta)^2:H^2(T^3)\to L^2(T^3)$. Finally, $\ti S(\theta)$ and $\hat T(\theta)=
\hat T_2(\theta)+\hat T_1(\theta)$
are defined respectively by (\re{tiHS}) and  (\re{K3}), (\re{K33}).

 However, 
the operator $A$ is not selfadjoint and even not symmetric,
which is a typical situation for the linearization of $U(1)$-invariant
nonlinear equations  \ci[Appendix B]{KK2014a}.
Respectively,  the  Bloch generators 
$\ti A(\theta)$ are not  selfadjoint in the Hilbert space
\be\la{XT3} 
\cX(T^3):=L^2(T^3)\oplus L^2(T^3)\oplus \C^3\oplus \C^3,\qquad  T^3:=\R^3/\Ga.
\ee
The main crux here  is that 
we cannot apply the von Neumann 
spectral theorem to the nonselfadjoint generators $A$ and $\ti A(\theta)$. 
We solve this problem by applying our spectral theory of the Hamilton operators 
with positive energy
\ci{KK2014a,KK2014b}, which  
is an
infinite-dimensional 
version of some Gohberg and Krein ideas 
from the  theory of parametric resonance
\ci[Chap. VI]{GK}.
This is why we need
the positivity of the energy operator $\ti B(\theta)$:
\be\la{Hpos2}
\cE(\theta,Y):=\langle Y, \ti B(\theta)Y\rangle_{T^3}\ge \vka(\theta)\Vert Y\Vert_{\cV(T^3)}^2,\qquad~~\mbox{\rm a.e.}~~\theta\in
\Pi^*\setminus\Ga^*,
\ee
where $\vka(\theta)>0$, 
the brackets  stand for the scalar product in $\cX(T^3)$, 
and we denote 
\be\la{cVi}
\cV(T^3):=H^1(T^3)\oplus H^1(T^3)\oplus \C^3 \oplus \C^3.
\ee
This positivity allows us to construct the spectral resolution 
of $\ti A(\theta)$ which 
implies the stability for
the linearized dynamics (\re{JDi}).
\medskip

The key result of the present paper is the proof of the positivity
(\re{Hpos2}) 
for the ions's charge  densities 
 $\si$ satisfying the following conditions on the corresponding Fourier transform $\ti\si(\xi)$.
 The first one is the Wiener-type condition  
\be\la{W1}
\qquad\mbox{\bf Wiener Condition:}\qquad\qquad  \Si(\theta):=\sum_m\Big[
 \fr{\xi\otimes\xi}{|\xi|^2}|\ti\si(\xi)|^2\Big]_{\xi=2\pi m-\theta}>0~,
 \quad  ~~{\rm a.e.}~~\theta\in \Pi^*\setminus\Ga^*.
\ee
This condition is
 an analog of Fermi Golden Rule
for crystals.
 The second condition reads
 \be\la{Wai}
\ti\si(2\pi m)=0,\quad m\in\Z^3\setminus 0.
\ee
%It provides the neutrality of the ground state:
%$\rho^0(x)\equiv 0$ and respectively, $\Phi^0(x)\equiv 0$.
The proof of the positivity (\re{Hpos2}) 
relies on a novel
factorization of the 
Hamilton functional.
%and a conversion of Sylvester's criterion.
This positivity 
 necessarily breaks down   at  $\theta\in\Ga^*$.
Examples \re{ex} and \re{ex2} demonstrate that 
the positivity  can break down at 
some  other points and submanifolds of $\Pi^*$.

Our main novelties are the following:
\medskip\\
I. The factorization of energy (\re{bb32}), (\re{bb33}) 
and (\re{fact}), (\re{sq}). 
\medskip\\
II. The energy bound from below (\re{bb}) for general densities $\si(x)$. 
\medskip\\
III.  The 
energy positivity   (\re{Hpos2}) under conditions (\re{W1}) and
(\re{Wai}) on $\si(x)$:
we show that the Wiener  condition  (\re{W1}) is necessary  
and sufficient for
the positivity (\re{Hpos2}) under assumption (\re{Wai}) (Theorem \re{tpose}).
\medskip\\
IV. An asymptotics of the ground state as $e\to 0$.
\medskip\\
V. An example of negative energy when the condition  (\re{Wai}) breaks down 
while the Wiener condition (\re{W1}) holds (Lemma \re{lne}).
\medskip\\
VI. Spectral resolution 
of nonselfadjoint Hamilton generators
%$A$ and $\ti A(\theta)$  
and  stability of the linearized dynamics. 
\br\la{remW}
The condition  (\re{Wai})  
cancels a negative contribution to
the energy, 
which is due to the electrostatic instability ("Earnshaw Theorem" \ci{Stratton}, see Remark \re{rT2}). 
\er
Let us comment on previous results in these directions.
\medskip\\
The crystal ground state 
for the Hartree-Fock equations 
was constructed by Catto, Le Bris, and  Lions  \ci{CBL2001,CBL2002}.
For the Thomas-Fermie model similar results were obtained in \ci{CBL1998}.
\medskip\\
The corresponding ground state 
in the Schr\"odinger-Poisson model was constructed in \ci{K2014}.   
The stability for the linearized dynamics 
was not established previously in any model. 
\medskip\\
In \ci{CS2012}, Canc\'es and Stoltz have established the well-posedness  for 
local perturbations of the ground state density matrix
in an infinite crystal 
for the reduced  Hartree-Fock model of crystal
in the  {\it random phase approximation}
with the Coulomb potential $w(x-y)=1/|x-y|$.
However, the  space-periodic nuclear potential
in the equation \ci[(3)]{CS2012}
does not depend on time that corresponds to 
the fixed ions's positions. Thus the back reaction of the electrons onto 
the nuclei is neglected.
\medskip\\
The nonlinear Hartree-Fock dynamics
for 
compact perturbations of the ground state
without the  random phase approximation
is not studied yet,
see the discussion in 
\ci{BL2005} and in the introductions of \ci{CLL2013,CS2012}.
\medskip\\
The paper \ci{CLL2013} deals with random reduced HF model of crystal  when 
the ions charge density and the electron density matrix are random processes,
and the action of the lattice translations on the probability space is ergodic.
The authors obtain suitable generalizations of the Hoffmann-Ostenhof 
and Lieb-Thirring inequalities  for ergodic density matrices, 
and
construt random potential which is a solution  to 
 the Poisson equation 
with the corresponding stationary stochastic  charge density. 
The main result is the  coincidence of this model with the thermodynamic limit in  
the case of the short range Yukawa interaction.
\medskip\\
In \ci{LS2014-1}, Lewin and Sabin  established the well-posedness for the 
reduced 
von Neumann equation 
with density matrices of infinite trace 
for pair-wise interaction potentials $w\in L^1(\R^3)$. The authors  also
proved  the asymptotic stability of the ground state 
for 2D crystals \ci{LS2014-2}.
Nevertheless, the case of
the Coulomb potential in 3D remains open.
\medskip\\
The spectral theory of the Schr\"odinger  operators
with space-periodic potentials 
is well developed, see \ci{RS4} and the references therein.
The scattering theory for short-range and long-range perturbations of
such operators was  constructed in \ci{GN1,GN2}.
\medskip

The plan of our paper is the following.
In Section 2  we recall our result \ci{K2014} on the existence of a ground state,
and in  Section 3 we establish small charge asymptotics of the ground state.
In Sections 4--6 we study the Hamilton structure
of the linearized dynamics and establish the energy bound from below.
In Section 7 we calculate the generator of the 
linearized dynamics 
in the  Fourier--Bloch representation.
In
Section 8 we prove  the  positivity of energy. 
In Section 9 we apply this positivity to
the stability of the linearized dynamics.
Finally, in Section 10 we construct examples of negative energy.
Appendices concern some technical calculations.
\medskip\\
{\bf Acknowledgments} The authors are grateful to Herbert Spohn for discussions and remarks.

%%%%%%%%%%%%%%%%%%%%%%%%%%%%%%%%%%%%%%%%%%%%%%%%%%%%%%%%%%%%%%%%%%%%%%%%%%%%%%%%%%%%%%%%%%%%%%%%%%%%%%
\setcounter{equation}{0}
\section{Space-periodic ground state}
%%%%%%%%%%%%%%%%%%%%%%%%%%%%%%%%%%%%%%%%%%%%%%%%%%%%%%%%%%%%%%%%%%%%%%%%%%%%%%%%%%%%%%%%%%%%%%%%%%%%%
Let us recall the results of \ci{K2014} on the existence of the ground state (\re{gr}).  
The Poisson equation (\re{LPS20}) for the $\Ga$-periodic potential $\Phi^0$
implies the neutrality of the periodic
cell $T^3=\R^3/\Ga$:
\be\la{neu}
 \int_{T^3} \rho^0(x)dx=0,
\ee
which is equivalent to the normalization condition
\be\la{nor}
 \int_{T^3} |\psi^0(x)|^2dx=Z
\ee
by (\re{ro+}).
We assume that
$ Z>0$,
since otherwise the theory is trivial.
The existence of the ground state (\re{gr}) is proved in \ci{K2014} under the condition 
\be\la{L12}
\si_{\rm per}(x):=\sum_n\sigma(x-n)\in  L^2(T^3).
\ee
The ion position $q^0\in T^3$ can be chosen arbitrary, and we will set
$
q^0=0.
$
\subsection{Minimization of energy per cell}
The wave function $\psi^0$ is constructed as a minimal point of the energy per cell
\be\la{U}
 U(\psi)=\fr12\int_{T^3}[|\na\psi(x)|^2+\rho(x)G_{\rm per}\rho(x)]dx,
\ee
where 
\be\la{U2}
\rho(x):=\si_{\rm per}(x)-e|\psi(x)|^2, 
\ee
while
the operator $G_{\rm per}:=-\De_{\rm per}^{-1}$ 
is defined by 
\be\la{Qc}
G_{\rm per}\vp(x)=\sum_{m\in \Z^3\setminus 0}e^{-i2\pi mx}\fr{\check\vp(m)}{|2\pi m|^2},
\qquad 
\check\vp(m)=\int_{T^3} e^{i2\pi mx}\vp(x)dx.
\ee
More precisely,
\be\la{min}
 U(\psi^0)=\min_{\psi\in \cM} U(\psi),
\ee
where $\cM$ denotes the manifold
\be\la{cM}
\cM:=\{ \psi\in H^1(T^3):\int_{T^3}|\psi(x)|^2dx=Z\}.
\ee
\subsection{Smoothness of the ground state}
The results  \ci{K2014} imply that
there exists a ground state with $\psi^0,\Phi^0\in H^2(T^3)$.
Hence
$\psi^0\Phi^0\in H^2(T^3)$, and the equation (\re{LPS10}) implies that 
\be\la{H4}
\psi^0\in H^4(T^3)\subset C^2(T^3).
\ee
In other words, 
\be\la{H42}
\psi^0(x)=\sum_{m\in \Z^3} \check\psi^0(m)e^{i2\pi m x},\qquad
\sum_{m\in \Z^3} \langle m\rangle^8|\check\psi^0(m)|^2<\infty,\qquad\langle m\rangle
:=(1+|m|^2)^{1/2}.
\ee

%%%%%%%%%%%%%%%%%%%%%%%
%%%%%%%%%%%%%%%%%%%%%%%

\setcounter{equation}{0}
\section{Small-charge asymptotics of the ground state}
%%%%%%%%%%%%%%%%%%%%%%%%%%%%%%%%%%%%%%%%%%%%%%%%%%%%%%%%%%%%%%%%%%%%%%%%%%%%%%%%%%%%%%%%%%%%%%%%%%%%%

We will 
need below
the asymptotics  as $e\to 0$ of the ground state (\re{gr})
corresponding 
to
a one-parametric family
of ion densities 
\be\la{fam}
\si(x)=e\mu(x)
\ee
with some fixed function $\mu\in L^2(\R^3)$. 
We   assume that
\be\la{muj}
\mu_{\rm per}(x):=\sum_{n\in \Z^3} \mu(x-n)\in L^2(T^3)
\ee
in accordance with  (\re{L12}).
Now the energy (\re{U}) reads
\be\la{Uc}
 U(\psi)=\fr12\int_{T^3}[|\na\psi(x)|^2
 +e^2\nu(x)G_{\rm per}\nu(x)]dx,\qquad \nu(x):=\mu_{\rm per}(x)-|\psi(x)|^2.
\ee
Denote by $\psi^0_e,\om^0_e$ the 
family of  ground states
with the parameter $e\in (0,1]$.
The energy (\re{Uc}) is obviously bounded
uniformly  in $e\in (0,1]$ for any fixed $\psi\in \cM$.
Hence, the energy of the minimizers   is also bounded
uniformly in $e\in (0,1]$.
In particular, 
the family   $\psi^0_e$ is bounded  in $H^1(T^3)$,
\be\la{psi}
\Vert\psi^0_e \Vert_{H^1(T^3)}\le C, \qquad e\in (0,1].
\ee
On the other hand, 
\be\la{L5}
\int\nu^0_e(x)G_{\rm per}\nu^0_e(x)dx\le C ,\qquad \nu^0_e(x):=
\mu_{\rm per}(x)-|\psi_e^0(x)|^2.
\ee
This estimate is
due to the uniform bound
\be\la{nu0}
\Vert\nu^0_e\Vert_{L^2(T^3)}\le C,\qquad  e\in (0,1]
\ee
which holds by (\re{muj}) and (\re{psi}).
Further, 
the equation (\re{LPS20})  reads
\be\la{LPS20e}
-\De\Phi^0_e(x)=e\nu^0_e(x).
\ee
We will choose the solution 
$\Phi^0_e=eG_{\rm per}\nu^0_e$, where the operator $G_{\rm per}$
is defined by (\re{Qc}). 
%Let us note that another choice of additive constant in $\Phi^0_e$
%reduces to the corresponding shift of $\om^0_e$, which does not 
%affect the operator (\re{Sop}).
The definition (\re{Qc}) implies the bound
\be\la{pot}
\Vert\Phi^0_e\Vert_{H^2(T^3)}\le 
e\Vert \nu^0_e\Vert_{L^2(T^3)}
\le
Ce, \qquad e\in (0,1]
\ee
by (\re{nu0}).
%%%%%%%%%%%%%%%%%%%%%%%%%%%%%%%%%%%%%%%%%%%%%%%%%%%%%%%%%%%%%%%%%%%%%%%%%%%%%%%%5
\begin{lemma}\la{llpos}
Let  condition  (\re{muj}) hold. Then for sufficiently  small $e>0$,
\be\la{Sope}
H^0_e:=-\fr12\De-e\Phi^0_e(x)-\om^0_e\ge 0,
\ee
and 
the ground state admits the following 
asymptotics as $e\to 0${\rm :}
\be\la{om0}
\om^0_e=\cO(e^2),
\ee
\be\la{p0}
\psi^0_e(x)=\ga_e+\chi_e(x),\qquad |\ga_e|^2=Z+\cO(e^4),
\qquad
\Vert\chi_e\Vert_{H^2(T^3)}= \cO(e^2).
\ee
\end{lemma}
%%%%%%%%%%%%%%%%%%%%%%%%%%%%%%%%%%%%%%%%%%%%%%%%%%%%%%%%%%%%%%%%%%%%%%%%%%%%%%%%%%%%%%%%%
\Pr  i) Equation (\re{LPS10}) reads
\be\la{LPS4r}
\om^0_e\psi^0_e(x)=-\fr12\De\psi^0_e(x)-e\Phi^0_e(x)\psi^0_e(x).
\ee
Hence, 
\be\la{om}
\om^0_e\langle\psi^0_e,\psi^0_e\rangle_{T^3}=\om^0_e Z=
\fr 12\langle\na\psi^0_e,
\na\psi^0_e\rangle_{T^3}-e\langle \Phi^0_e\psi^0_e,\psi^0_e\rangle_{T^3}, 
\ee
which implies the uniform bound
\be\la{ub}
|\om^0_e|\le C<\infty,\qquad e\in(0,1]
\ee
by (\re{nor}),  (\re{psi}) and (\re{pot}).
Moreover, (\re{LPS4r}) and (\re{pot}) suggest that 
$\om^0_e$ is close to an eigenvalue of $-\fr12\De$:
\be\la{eig}
\om^0_e\approx |2\pi k|^2
\ee
with some $k\in\Z^3$. Indeed, (\re{LPS4r}) can be rewritten as 
\be\la{re}
(\fr 12 |2\pi m|^2-\om^0_e)\check\psi^0_e(m)=
\check {r_e}(m),\qquad r_e:=e\Phi^0_e\psi^0_e
\ee
and hence, 
\be\la{re2}
\sum_{m\in\Z^3} (\fr 12 |2\pi m|^2-\om^0_e)^2|\check\psi^0_e(m)|^2=\cO(e^4),
\ee
since $\Vert r_e\Vert_{L^2(T^3)}=\cO(e^2)$ by (\re{pot}).
Denote by $\lam_e$ the value of $|2\pi m|^2$ corresponding to 
the minimal
magnitude of $|\fr 12 |2\pi m|^2-\om^0_e|$. Then  (\re{re2}) implies that
\be\la{re5}
\sum_{|2\pi m|^2\ne\lam_e} |\check\psi^0_e(m)|^2
=\cO(e^4),
\ee
since the set of possible values of $\fr12|2\pi m|^2-\om^0_e$ is discrete and 
possible values of $\om^0_e$ are bounded by (\re{ub}).
Moreover,  (\re{re2}) can be rewritten as
\be\la{re2'}
(\fr12\lam_e-\om^0_e)^2Z
+\sum_{|2\pi m|^2\ne  \lam_e}\Big[(\fr 12 |2\pi m|^2-\om^0_e)^2-(\fr12\lam_e-\om^0_e)^2\Big]
|\check\psi^0_e(m)|^2
=\cO(e^4)
\ee
since 
\be\la{re4}
\sum_{m\in\Z^3} |\check\psi^0_e(m)|^2= Z
\ee
due to the normalization (\re{nor}).
Hence,
\be\la{re6}
|\fr12\lam_e-\om^0_e|=\cO(e^2),
\ee
since the sum in (\re{re2'}) is nonnegative.
Let us show that
(\re{re2'}) also implies that
\be\la{re6'}
\sum_{|2\pi m|^2\ne  \lam_e}(|2\pi m|^2-\lam_e)^2|\check\psi^0_e(m)|^2=\cO(e^4).
\ee
First, (\re{re2'}) gives that
$$
\sum_{|2\pi m|^2\ne  \lam_e}
(|2\pi m|^2-\lam_e)(\fr 12 |2\pi m|^2+\fr12\lam_e-2\om^0_e)
|\check\psi^0_e(m)|^2
=\cO(e^4)
$$
However,
$
2\om^0_e=\lam_e+\cO(e^2)
$
by (3.21). Hence,
$$
\sum_{|2\pi m|^2\ne  \lam_e}
(|2\pi m|^2-\lam_e)(|2\pi m|^2-\lam_e+\cO(e^2))
|\check\psi^0_e(m)|^2
=\cO(e^4).
$$
Now (\re{re6'}) follows from (\re{re5}) 
since $\lam_e$ is bounded for small $e>0$
by (\re{re6}) and (\re{ub}).
\medskip\\
%%%%%%%%%%%%%%%%%%%%%%%%
ii) Now let us  prove that $\lam_e=0$ for small $e>0$. Indeed, 
the energy of the ground state reads 
\be\la{eg}
U(\psi^0_e)=\fr12\sum_{m\in\Z^3} |2\pi m|^2|\check\psi^0_e(m)|^2+\cO(e^2)
\ee
by (\re{Uc}) and  (\re{L5}). 
On the other hand, (\ref{re6'}) implies
\be\la{eg2}
\sum_m |2\pi m|^2|\check\psi^0_e(m)|^2
=\lam_e Z+\sum_{|2\pi m|^2\ne  \lam_e} (|2\pi m|^2-\lam_e)|\check\psi^0_e(m)|^2
=\lam_e Z+\cO(e^4).
\ee
Substituting (\re{eg2}) into (\re{eg}), we obtain 
\be\la{eg3}
U(\psi^0_e)
=\fr12\lam_e Z+\cO(e^2),\quad \lam_e\ge 0.
\ee
On the other hand, taking $\psi(x)\equiv \sqrt{Z}$,  we ensure 
that the energy minimum (\re{min})  does not exceed $\cO(e^2)$.
Hence, (\re{eg3}) implies that $\lam_e=0$ for small $e>0$, 
since the set of all possible values of $\lam_e Z$ is discrete.
Therefore, (\re{om0}) holds by  (\re{re6}).
%%%%%%%%%%%%%%%%%%  lena %%%%%%%%%%%%%%%%%%%%%%%%%%%%%%%%%%%%
\medskip\\
iii) Now we can prove the asymptotics (\re{p0}). Namely, the 
first identity holds if we set
\be\la{ho}
\ga_e=\check\psi^0_e(0),\qquad\chi_e(x)=
\sum_{m\ne 0}  e^{-i2\pi mx}\check\psi^0_e(m).
\ee
Then  the 
second asymptotics of  (\re{p0}) holds
by
(\re{re4}) and  (\re{re5}) with $\lam_e=0$.
The last asymptotics  of  (\re{p0}) holds since
\be\la{re22}
\sum_{m\ne 0} |2\pi m|^4|\check\psi^0_e(m)|^2=\cO(e^4)
\ee
due to (\re{re6'}) with $\lam_e=0$. 
Finally, (\re{pot}) and (\re{om0})  with small $e>0$
imply that the lowest eigenvalue of the Schr\"odinger operator 
$H^0_e$
 in $L^2(T^3)$ is close to zero. Hence,  its zero eigenvalue 
is exactly the lowest eigenvalue,
since the spectrum of this operator is discrete.
Therefore, the nonnegativity   (\re{Sope}) is proved for small $e>0$.
\bo

%%%%%%%%%%%%%%%%%%%%%%%%%%%%%%%%%%%%%%%%%%%%%%%%%%%%%%%%%%%%%%%%%%%%%%%%%%%%%%%%%
\setcounter{equation}{0}
\section{Linearized dynamics}
%%%%%%%%%%%%%%%%%%%%%%%%%%%%%%%%%%%%%%%%%%%%%%%%%%%%%%%%%%%%%%%%%%%%%%%%%%%%%%%%%
Let us consider the linearized system  (\re{LPS1Li}).
We recall that $G:=-\De^{-1}$. 
The meaning of the terms with $G$ will be adjusted below, see 
Lemma \re{lDB}. We assume further that (\re{L12}) holds,
and additionally, 
\be\la{L123i}
\langle x\rangle^2\si\in L^2(\R^3),\qquad (\De-1)\si\in L^1(\R^3).
\ee
For $f(x)\in C_0^\infty(\R^3)$ the Fourier transform is defined by
\be\la{Fu2}
f(x)=\fr 1{(2\pi)^3}\int_{\R^3} e^{-i\xi x}\ti f(\xi)d\xi,\qquad x\in \R^3;
\qquad
\ti f(\xi)=\int_{\R^3} e^{i\xi x}f(x)dx,\qquad \xi\in \R^3.
\ee
The  conditions (\re{L123i}) imply that
\be\la{L123}
(\De-1)\ti\si\in L^2(\R^3),\qquad \langle\xi\rangle^2\ti\si(\xi)\le \const.
\ee
%%%%%%%%%%%%%%%%%%%%%%%%%%%%%%%%%%%%%%%%%%%%%%%%%%%%%%%%%%%%%%
We consider the case when the ground state $\psi^0(x)$ can be 
taken to be a
 real function.
Then   (\re{LPS1Li})--(\re{realdot}) imply that 
the operator-matrix $A$
 is given by  (\re{JDi})
where 
$ S$  denotes the operator with the ``matrix'' 
\be\la{S}
 S(x,n):=e\psi^0(x)G\na\si(x-n):~~n\in\Z^3,~x\in\R^3.
\ee
Finally, $T$ is the real matrix with entries
\be\la{T}
T(n,n'):=-\ds\langle  G\na\otimes\na\si(x-n'),  \si(x-n) \rangle
+  \ds\langle\Phi^0,\na\otimes \na\si\rangle\de_{nn'}
=T_1(n-n')+T_2(n-n').
\ee
%%%%%%%%%%%%%%%%%%%%%%%%%%%%%%%%%%%%%%%%%%%%%%%%%%%%%%%%%%%%%%%%%%%%%%%%%%%%%%%%%%
The operators 
$G\psi^0: L^2(\R^3)\to L^2(\R^3)$ and $S:l^2_3:=l^2_3(\Z^3)\otimes \C^3\to L^2(\R^3)$ 
are not bounded due to the ``infrared divergence", see Remark \re{run}.
In the next section, we will construct a dense domain for all these operators. 

On the other hand, the corresponding 
operators $T_1$  and $T_2$ are bounded by the following lemma.
Denote by $\Pi$  the primitive cell 
\be\la{cellPi}
\Pi:=\{(x_1,x_2,x_3):0\le x_k\le 1,~k=1,2,3\}.
\ee
Let us define the Fourier transform on $l^2_3$ as
\be\la{Fu}
\hat Q(\theta)=\sum_{n\in\Z^3} e^{in\theta}Q(n),\qquad~~{\rm a.e.}~~ \theta\in \Pi^*;\qquad
Q(n)=\fr 1{|\Pi^*|}\int_{\Pi^*}   e^{-in\theta}\hat Q(\theta)d\theta,~~n\in\Z^3,
\ee
where  $\Pi^*=2\pi \Pi$ denotes the primitive cell of the lattice $\Ga^*$ and 
the series converges in $L^2(\Pi^*)$.
%%%%%%%%%%%%%%%%%%%%%%%%%%%%%%%%%%%%%%%%%%%%%%%%%%%%%%%%%%%%%%%%%%%%%%%%%%%%%%%%%%%%%%%
\begin{lemma}\la{lT}
The operators $T_1$ and $T_2$ are bounded in  $l^2_3$ under  condition (\re{L123i}).
\end{lemma}
%%%%%%%%%%%%%%%%%%%%%%%%%%%%%%%%%%%%%%%%%%%%%%%%%%%%%%%%%%%%%%%%%%%%%%%%%%%%%%%%%%%%%%%%%%
\Pr
The first operator $T_1$ reads as the convolution: 
$T_1 Q(n)=\sum T_1(n-n')Q(n')$,
where
\be\la{Kn}
T_1(n)=-\langle \na\otimes G\na\si(x),   \si(x-n) \rangle.
\ee
In the Fourier transform (\re{Fu}), the convolution operator $T_1$ becomes the 
multiplication,
\be\la{K22}
\widehat{T_1 Q}(\theta)=\hat T_1(\theta) \hat Q(\theta  ),\qquad ~~{\rm a.e.}~~\theta\in\Pi^*\setminus\Ga^*.
\ee
By the Parseval identity, it suffices to check that
the ``symbol" $\hat T_1(\theta)$ is a bounded function. This follows by direct 
calculation from (\re{T}). First, we apply the Parseval identity: 
\beqn\la{K3}
\hat T_1(\theta)\!\!&\!=\!&\!\!
-\sum_n e^{in\theta}\langle  \na\otimes G\na\si(x),\si(x-n)\rangle
=
\fr1{(2\pi)^3}\sum_n e^{in\theta}\langle  \fr{\xi\otimes\xi}{|\xi|^2}
\ti\si(\xi),\ti\si(\xi)e^{in\xi}\rangle
\nonumber\\
\nonumber\\
\!\!&\!=\!&\!\!
\fr 1{(2\pi)^{3}}\langle 
\fr{\xi\otimes\xi}{|\xi|^2}\ti\si(\xi), \ti\si(\xi)
\sum_n
e^{in(\theta+\xi)}\rangle
=
\sum_m
\Big[\fr{\xi\otimes\xi}{|\xi|^2}|\ti\si(\xi)|^2\Big]_{\xi=2\pi m-\theta}~,
\quad\theta\in\Pi^*\setminus\Ga^* 
\eeqn
since the  sum over $n$ equals $\ds|\Pi^*|\sum_m \de(\theta+\xi-2\pi m)$
by the Poisson summation formula \ci{Her1}. Finally, $|\ti\si(\xi)|\le C\langle 
\xi\rangle^{-2}$ by (\re{L123}). Hence,
\be\la{K4}
|\hat T_1(\theta)|\le C_1\sum_m|\ti\si(2\pi m-\theta)
\ti\si(2\pi m-\theta)|\le C_2\sum_m\langle m\rangle^{-4}<\infty.
\ee
ii)
Finally, 
\be\la{K32}
\widehat{T_2Q}(\theta)=\hat T_2\hat Q(\theta),\qquad\theta\in\Pi^*,
\ee
where
\be\la{K33}
\hat T_2=\langle \Phi^0(x), \na\otimes\na\si(x)\rangle
%=-\fr1{(2\pi)^3}\langle \ti\rho^0(\xi) \fr{\xi\otimes\xi}{|\xi|^2},
%\ti\si(\xi)\rangle   
\ee
by  (\re{LPS20}). 
The expression 
is finite  by (\re{L123i}), since $\Phi^0\in H^2(T^3)$
is a bounded periodic function.\bo

%%%%%%%%%%%%%%%%%%%%%%%%%%%%%%%%%%%%%%%%%%%%%%%%%%%%%%%%%%%%%%%%%%%%%%%%%%%%%%%%%%%%%%%%%%%

%%%%%%%%%%%%%%%%%%%%%%%%%%%%%%%%%%%%%%%%%%%%%%%%%%%%%%%%%%%%%%%%%%%%%%%%%%%%%%%%%%%%%%%%%%%%
%%%%%%%%%%%%%%%%%%%%%%%%%%%%%%%%%%%%%%%%%%%%%%%%%%%%%%%%%%%%%%%%%%%%%%%%%%%%%%%%%%%%%%%%%%%%%%%
\setcounter{equation}{0}
\section{The Hamilton structure and the domain}
%%%%%%%%%%%%%%%%%%%%%%%%%%%%%%%%%%%%%%%%%%%%%%%%%%%%%%%%%%%%%%%%%%%%%%%%%%%%%%%%%%%%%%%%%%%%%%%
To construct solutions of the system (\re{JDi}), we need to diagonalize its generator $A$.
The main problem is that this generator is neither selfadjoint and even not symmetric, 
so we cannot apply the von Neumann 
spectral theorem. We will 
solve this problem by applying our spectral theory of  Hamilton operators 
with positive energy \ci{KK2014a,KK2014b} to the Bloch representation of $A$.

In this section 
we study  the domain  of the generator $A$.  Denote
\be\la{cV}
 \cV:=H^1(\R^3)\oplus H^1(\R^3) \oplus l^2_{3}\oplus l^2_{3},\qquad l^2_{3}:=l^2(\Z^3)\otimes \C^{3}.
\ee
It is easy to check that
the Hamilton representation (\re{AJB}) formally holds
with the   symplectic matrix  
\be\la{J}
J=\left(\ba{cccc}   
0  & \fr12 &  0 & 0\\
-\fr12 & 0 &  0 & 0\\
0 & 0 &  0 & 1\\
0 & 0 & -1 & 0
\ea\right).
\ee
\begin{definition}\la{dD}
i) $\cS_+:= \cup_{\ve>0} \cS_\ve$, where $\cS_\ve$  is the space of functions 
$\vp\in\cS(\R^3)$, whose Fourier transforms  $\hat\vp(\xi)$ vanish in the $\ve$-neighborhood of the lattice $\Ga^*$,
\medskip\\
ii) $l_c=\cup_{R\in\N}l_c(R)$, where $l_c(R):=\{Q\in l^2_3: Q(n)=0 ~~ {\rm for} ~~ |n|>R\}$.
\medskip\\
iii)  
$  \cD:=\{Y=(\Psi_1,\Psi_2,Q,P)\in\cX:  \Psi_1,\Psi_2\in \cS_+,~~~ Q\in l_c,~~~ P\in l_c \}.$
\end{definition}
Obviously, $\cD$ is dense in $\cX$.

%%%%%%%%%%%%%%%%%%%%%%%%%%%%%%%%%%%%%%%%%%%%%%%%%%%%%%%%%%%%%%%%%%%%%%%%%%
\begin{theorem}\la{tHam}
Let conditions (\re{L123i}) hold. Then
$B$ is a symmetric operator on 
the domain $\cD\subset\cX$.
\end{theorem}
%%%%%%%%%%%%%%%%%%%%%%%%%%%%%%%%%%%%%%%%%%%%%%%%%%%%%%%%%%%%%%%%%%%%%%%%%%%%%%%%%%%%%%%%%%%%
\Pr
%%%%%%%%%%%%%%%%%%%%%%%%%%%%%%%%%%%%%%%%%%%
Formally the matrix  (\re{AJB}) is symmetric. 
The following lemma implies that 
$B$ is defined on $\cD$.
%%%%%%%%%%%%%%%%%%%%%%%%%%%%%%%%%%%%%%%%%%%%%%%%%%%%%%%%%%%%%%%%%%%%%%%%%%%%%
\begin{lemma}\la{lDB}
i)  
$\psi^0 G\psi^0 \vp\in L^2(\R^3)$  and $S^*\vp\in l^2_3$  for 
$\vp\in\cS_+$.
\medskip\\
ii) $S Q\in L^2(\R^3)$ for $Q\in l^c$.
\end{lemma}
%%%%%%%%%%%%%%%%%%%%%%%%%%%%%%%%%%%%%%%%%%%%%%%%%%%%%%%%%%%%%%%%%%%%%%%%%%%%%%%
\Pr i)
First, note that
\be\la{Qp}
 G\psi^0\vp=F^{-1}\fr{[\ti\psi^0*\ti\vp](\xi)}{|\xi|^2}.
\ee
Further,  $\ti\psi^0(\xi)=(2\pi)^3\sum_{m\in\Z^3} \check\psi^0(m)\de(\xi-2\pi m)$. 
Respectively,
\be\la{conv}
 [\ti\psi^0*\ti\vp](\xi)= (2\pi)^3\sum _{m\in \Z^3}
 \check\psi^0(m)\hat\vp(\xi-2\pi m)=0,
 \qquad |\xi|<\ve
\ee
if $\vp\in\cS_\ve$ with some $\ve>0$. Moreover, 
$\ti\psi^0*\ti\vp\in L^2(\R^3)$, since $\psi^0\vp\in L^2(\R^3)$.
Hence,  $\vp$ belongs to the domain of $G\psi^0$ and of $\psi^0G\psi^0$.
\medskip\\
Now consider $S^*\vp$. Applying (\re{S}), the Parseval identity and  (\re{conv}) we get
for $\vp\in\cS_\ve$
\beqn\la{conv2}
[S^*\vp](n)&=&e\int \psi^0(x)\vp(x)G\na\si(x-n)dx
=e\langle\psi^0(x)\vp(x),G\na\si(x-n)\rangle
\nonumber\\
\nonumber\\
&=&\frac{ie}{(2\pi)^{3}}\int_{|\xi|>\ve}[\ti\psi^0*\ti\vp](\xi)
\fr{\xi\ov{\ti\si}(\xi)e^{-in\xi}}{|\xi|^2}d\xi.
\eeqn
Here 
$\pa^\al[\ti\psi^0*\ti\vp](\xi)\langle\xi\rangle^4\in L^2(\R^3)$ 
for all $\al$ by  (\re{H42}), since $\ti\vp\in\cS(\R^3)$.
Moreover, $\pa^\al\ti\si\in L^2(\R^3)$ for $|\al|\le 2$ by (\re{L123}).
Hence, integrating by parts twice, and taking into account (\re{conv}), we obtain 
\be\la{conv4}
 |[S^*\vp](n)|\le C\langle n\rangle^{-2},
\ee
which implies that $S^*\vp\in l^2_3$.
\medskip\\
ii) Let us  check that $SQ\in L^2(\R^3)$ for $Q\in l_c$. 
The Fourier transform of $SQ$ reads as
\beqn\la{SF}
\widetilde{SQ}(\xi)&=&
e F_{x\to\xi}\sum_n \psi^0(x)G\na\si(x-n)Q(n)
=e\sum_n \ti\psi^0*F_{x\to\xi}[G\na\si(x-n)]Q(n)
\nonumber\\
\nonumber\\
&=&e(2\pi)^3
\int \sum_m \check\psi^0(m) \de(\eta-2\pi m) \widetilde{G\na\si}(\xi-\eta) 
\sum_{n}e^{in(\xi-\eta)}Q(n)d\eta
\nonumber\\
\nonumber\\
&=&e(2\pi)^3
\sum_m \check\psi^0(m)  \widetilde{G\na\si}(\xi-2\pi m) 
\ti Q(\xi-2\pi m).
\eeqn
Hence, the Parseval identity gives that
\be\la{SF2}
\Vert SQ\Vert_{L^2(\R^3)}=C\Vert \widetilde{SQ}\Vert_{L^2(\R^3)}
\le C_1\sum_m |\check\psi^0(m)|  
\Vert\widetilde{G\na\si}(\xi)\ti Q(\xi)\Vert _{L^2(\R^3)}
\ee
It remains to note that the sum over $m$ is finite by  (\re{H42})
because
\be\la{SF3}
\Vert\widetilde{G\na\si} \ti Q\Vert_{L^2(\R^3)}^2=
\int \fr{1}{|\xi|^2}|\ti\si(\xi) \ti Q(\xi)|^2d\xi
\le C(Q)\int \fr{|\ti\si(\xi)|^2}{|\xi|^2}d\xi
\ee
since the function $\ti Q(\xi)$ is bounded for $Q\in l_c$.
Finally, the last integral is finite by  (\re{L123}).
\bo
%%%%%%%%%%%%%%%%%%%%%%%%%%%%%%%%%%%%%%%%%%%%%%%%%%%%%%%%%%%%%%%%%
\medskip\\
This lemma implies that $BY\in\cX$ for $Y\in \cD$.
The symmetry of $B$ on $\cD$ is evident from (\re{AJB}).
Theorem \re{tHam} is proved. \bo
%%%%%%%%%%%%%%%%%%%%%%%%%%%%%%%%%%%%%%%%%%%%%%%%%%%%%%%%%%%%%%%
\begin{remark}\la{run}
The infrared singularity at $\xi=0$ of the integrands (\re{Qp}), (\re{conv2}) and (\re{SF3}) 
demonstrates that all operators
$G\psi^0:L^2(\R^3)\to L^2(\R^3)$, $S:l^2_3\to L^2(\R^3)$ and 
$S^*:L^2(\R^3)\to l^2_3$ are unbounded.
\end{remark}
\bc\la{cAA*}
The proof of 
Theorem \re{tHam} shows that the operator $A$ is defined on $\cD$, as well as 
the "formal adjoint" $A^*$, which is  defined by the identity
\be\la{A*}
\langle AY_1,Y_2\rangle=\langle Y_1,A^*Y_2\rangle,\qquad Y_1,Y_2\in \cD.
\ee
\ec

%%%%%%%%%%%%%%%%%%%%%%%%%%%%%%%%%%%%%%%%%%%%%%%%%%%%%%%%%%%%%%%%%%%%%%%%%%%%%%%%%%%%%%%%%%%%
%%%%%%%%%%%%%%%%%%%%%%%%%%%%%%%%%%%%%%%%%%%%%%%%%%%%%%%%%%%%%%%%%%%%%%%%%%%%%%%%%%%%%%%%%%%%%%%
\setcounter{equation}{0}
\section{Factorization of energy  and bound from below}
%%%%%%%%%%%%%%%%%%%%%%%%%%%%%%%%%%%%%%%%%%%%%%%%%%%%%%%%%%%%%%%%%%%%%%%%%%%%%%%%%%%%%%%%%%%%
%%%%%%%%%%%%%%%%%%%%%%%%%%%%%%%%%%%%%%%%%%%%%%%%%%%%%%%%%%%%%%%%%%%%%%%%%%%%%%%%%%%%%%%%%%%%%

The  equation (\re{JDi})  is
formally  a Hamiltonian system with  Hamilton functional 
$\fr12\langle Y,B Y \rangle$.
Next theorem means the stability  property of the linearized crystal. 
%%%%%%%%%%%%%%%%%%%%%%%%%%%%%%%%%%%%%%%%%%%%%%%%%%%%%%%%%%%%%%%%%%%%%%%%%%%%%%%
\begin{theorem}\la{tpos}
Let conditions (\re{L123i}) hold. Then the operator $B$ on the domain $\cD$ is bounded from below:
\be\la{bb}
 \langle Y,BY \rangle\ge -C\Vert Y\Vert^2_{\cX},~~~~~~~Y\in \cD.
\ee
\end{theorem}
%%%%%%%%%%%%%%%%%%%%%%%%%%%%%%%%%%%%%%%%%%%%%%%%%%%%%%%%%%%%%%%%%%%%%%%%%%%
\Pr 
For  $Y=(\Psi_1,\Psi_2,Q,P)\in \cD$  the quadratic form reads (with the notations 
(\re{S})--(\re{T}))
\beqn\la{bb2}
\!\!\!\!\langle Y,BY \rangle\!\!&\!\!=\!\!&\!\!2\sum_j\langle \Psi_j,H^0\Psi_j \rangle
\!+\!4e^2\langle \psi^0\Psi_1,G\psi^0\Psi_1\rangle
\!+\!2[\langle   \Psi_1, S Q \rangle+\langle   Q, S^{\5*} \Psi_1 \rangle]
\!+\!\langle Q,T_1Q\rangle
\nonumber\\
\nonumber\\
&\!\!\!&\!\!\!+\langle Q,T_2Q\rangle
\!+\!\langle P,M^{-1}P\rangle.
\eeqn
Here the first sum is bounded from below,
the operator $T_2$ is bounded in $l^2_3$ by Lemma \re{lT}, while the operator $M^{-1}$ is positive.  
Our basic observation is that 
\be\la{bb3}
 \beta(\Psi_1,Q):=
 4e^2\langle \psi^0\Psi_1,G\psi^0\Psi_1\rangle
+2[\langle   \Psi_1, S Q \rangle+\langle   Q, S^{\5*} \Psi_1 \rangle]
+\langle Q,T_1Q\rangle\ge 0.
\ee
Indeed, the operators factorize as follows:
\be\la{bb32}
 e^2\psi^0 G\psi^0=f^*f, \qquad S=f^*g, \qquad T_1= g^*g, 
\ee
where 
\be\la{fg}
f:=e\sqrt{G}\psi^0,\qquad  g(x,n)={\na}{\sqrt{G}}\si(x-n).
\ee
%and 
%$g$ is the operator with the matrix entries
%$g(x,n)={\na}{\sqrt{G}}\si(x-n)$.
Then the quadratic form   (\re{bb3}) becomes the "perfect square" 
\be\la{bb33}
~~~~~~~~~~~~~~~~~~~~~~~~~~~~  \beta(\Psi,Q)=
 \langle 2f \Psi_1 +gQ, 
 2f  \Psi_1 +gQ \rangle\ge 0.~~~~~~~~~~~~~~~~~~~~~~~~~~~~~~~~~~~~~~~~~~~~~~~~~~~~~~~~~\bo
\ee
%%%%%%%%%%%%%%%%%%%%%%%%%%%%%%%%%%%%%%%%%%%%%%%%%%%%%%%%%%%%%%%%%%%%%%%%%%%%%%%%%%%%%%%5
\bc
The operator $B$ admits  selfadjoint extensions by the Friedrichs Theorem \ci{RS2}.
\ec

%%%%%%%%%%%%%%%%%%%%%%%%%%%%%%%%%%%%%%%%%%%%%%%%%%%%%%%%%%%%%%%%%%%%%%%%%%%%%%
%%%%%%%%%%%%%%%%%%%%%%%%%%%%%%%%%%%%%%%%%%%%%%%%%%%%%%%%%%%%%%%%%%%%%%%%%%%%%%%
\setcounter{equation}{0}
\section{Generator in  the Fourier--Bloch transform}
%%%%%%%%%%%%%%%%%%%%%%%%%%%%%%%%%%%%%%%%%%%%%%%%%%%%%%%%%%%%%%%%%%%%%%%%%%%%%%%%%
%%%%%%%%%%%%%%%%%%%%%%%%%%%%%%%%%%%%%%%%%%%%%%%%%%%%%%%%%%%%%%%%%%%%%%%%%%%%%%%
We reduce the operators  $A=JB$ and $K$ by
the  Fourier--Bloch--Gelfand--Zak transform \ci{DK2005,PST}.

\subsection{The discrete Fourier transform}
Let us consider a vector 
$  Y=(\Psi_1, \Psi_2, Q, P)\in \cX$, and denote 
\be\la{Yn2}
 Y(n)=(\Psi_1(n,\cdot),\Psi_2(n,\cdot), Q(n),P(n))~,~~~~~~n\in\Z^3,
\ee
where 
\be\la{YP}
\Psi_j(n,y)=
\left\{
\ba{ll}
\Psi_j(n+y),&~~{\rm a.e.}~~y\in \Pi,\\
0,&y\not\in \Pi.
\ea
\right.
\ee
Obviously, $Y(n)$ with different $n\in\Z^3$ are orthogonal vectors in $\cX$, and 
\be\la{ort}
Y=\sum_n Y(n),
\ee
where the sum converges in $\cX$. The norms in $\cX$ and $\cV$ can be represented as
\be\la{nV}
 \Vert Y\Vert_\cX^2=\sum_{n\in\Z^3}\Vert Y(n)\Vert_{\cX(\Pi)}^2,\quad\quad
 \Vert Y\Vert_\cV^2=\sum_{n\in\Z^3}\Vert Y(n)\Vert_{\cV(\Pi)}^2,
\ee
where 
\be\la{KV}
 \cX(\Pi):=L^2(\Pi)\oplus L^2(\Pi) \oplus \C^3\oplus \C^3,\quad\quad
 \cV(\Pi):=H^1(\Pi)\oplus H^1(\Pi) \oplus \C^3\oplus \C^3.
\ee
Further,
the ground state (\re{gr}) is invariant with respect to translations of the lattice $\Ga$, 
and hence the operator  $A$ commutes with these translations. Namely, (\ref{S}) implies that
\be\la{S2}
 S (x,n)=S(x-n,0),  
\ee
since 
%the operator $G$ commutes with translations of $\R^3$ and
$\psi^0(x)$ is a $\Ga$-periodic function. Similarly,  (\ref{T})  implies
that $T$ commutes with  translations of $\Ga$.
%%%%%%%%%%%%%%%%%%%%%%%%%%%%%%%%%%%%%%%%%%%%%%%%%%%%%%%%%%%%%%%%%%%%%%%%%%%%%%%%%%%%%%
Hence,
$A$ can be reduced by the discrete Fourier transform.
Namely,  applying the Fourier transform
$F_{n\to\theta}$ to the function $Y(\cdot)$ from (\re{Yn2}), we obtain
\be\la{F}
 \hat Y(\theta)=F_{n\to \theta}Y(n):=\sum\limits_{n\in \Z^3}e^{in\theta}Y(n)
 =(\hat{\Psi}_1(\theta,\cdot),\hat{\Psi}_2(\theta,\cdot),\hat {Q}(\theta),\hat{P}(\theta)),
 \quad~~{\rm a.e.}~~\theta\in \R^3,
\ee
where
\be\la{FPsi}
 \hat \Psi_j(\theta,y)=\sum\limits_{n\in \Z^3}e^{in\theta}\Psi_j(n+y),
 \quad~~{\rm a.e.}~~\theta\in \R^3,\quad ~~{\rm a.e.}~~y\in\R^3.
\ee
The function $\hat Y(\theta)$ is $\Ga^*$-periodic in $\theta$.
The series (\re{F}) converges in $L^2(\Pi^*,\cX(\Pi))$, since the series
(\re{ort}) converges in $\cX$.
The inversion formula  is given by 
\be\la{FI}
  Y(n)=|\Pi^*|^{-1}\int_{\Pi^*}e^{-in\theta}\hat Y(\theta)d\theta
\ee
(cf. (\re{Fu})).
The Parseval--Plancherel identity gives  
\be\la{PP}
 \Vert Y\Vert_\cV^2=|\Pi^*|^{-1}\Vert\hat Y\Vert_{L^2(\Pi^*,\cV(\Pi))}^2,
\qquad
 \Vert Y\Vert_\cX^2=|\Pi^*|^{-1}\Vert\hat Y\Vert_{L^2(\Pi^*,\cX(\Pi))}^2.
\ee
The functions $\hat {\Psi}_j(\theta,y)$ are  $\Ga$-quasiperiodic in $y$; i.e.,
\be\la{qp}
 \hat{\Psi}_j(\theta,y+m)=e^{-im\theta}\hat{\Psi}_j
(\theta,y), \qquad m\in\Z^3.
\ee

%%%%%%%%%%%%%%%%%%%%%%%%%%%%%%%%%%%%%%%%%%%%%%%%%%%%%%%%%%%%%%%%%%%
\subsection{Generator in the discrete Fourier transform}
Let us consider $Y\in\cD$ and calculate 
the Fourier transform (\ref{F}) for $AY$.  
Using  (\ref{T}), (\ref{conv2}), (\ref{S2}), and taking into account the $\Ga$-periodicity of $\Phi^0(x)$ and  $\psi^0(x)$,
we obtain that
\be\la{CPFh}
 \widehat{AY}(\theta)=\hat A(\theta){\hat Y}(\theta),\qquad
~~{\rm a.e.}~~\theta\in\R^3\setminus \Ga^*,
\ee
where $\hat A(\theta)$ is a  $\Ga^*$-periodic operator function,
\be\la{tiAh}
 \hat A(\theta)=\left(\!\ba{cccl}
 0 & H^0 & 0 &  0\medskip\\
 -H^0-2e^2\psi^0\hat G(\theta)\psi^0  & 0  &
 \hat S(\theta) &  0\\
 0  &  0 &  0 & M^{-1}\\
 -2\hat S^{\5*}(\theta)&0  &-\hat T(\theta)&0\\
\ea\!\right).
\ee
by (\re{JDi}) and (\re{AJB}).
Here 
\be\la{qp4}
\hat G(\theta)\hat\vp(\theta,y)=\sum_m\fr{\check\vp(\theta,m)}{(2\pi m-\theta)^2}
e^{i(2\pi m-\theta)y},\qquad~~{\rm a.e.}~~
\theta\in\R^3\setminus \Ga^*.
\ee
This expression  is well-defined for $\vp(x)=\psi^0(x)\Psi_1(x)$ with
$\Psi_1\in\cS_\ve$ since 
\be\la{qp5}
\check\vp(\theta,m)=\ti\vp(2\pi m-\theta)=0
\quad {\rm for}\quad|2\pi m-\theta|<\ve
\ee
according to (\re{conv}).

%%%%%%%%%%%%%%%%%%%%%%%%%%%%%%%%%%%%%%%%%%%%%%%%%%%%%%%%%%%%%%%%%%%%%%%%
\begin{lemma}\label{SST}
The operator $\hat S(\theta)$ acts as follows:
\begin{equation}\label{S-act}
 \hat S(\theta)\hat { Q}(\theta)
 =\hat S(\theta)\hat Q(\theta),
 \quad {\rm where}\quad \hat S(\theta)=e\psi^0\hat G(\theta)\na\hat\sigma(\theta,y).
\end{equation}
\end{lemma}
%%%%%%%%%%%%%%%%%%%%%%%%%%%%%%%%%%%%%%%%%%%%%%%%%%%%%%%%%%%%%%%%%%%%
\begin{proof}
For $x=y+n$  equations \eqref{ro+2} and \eqref{S} imply
\beqn\nonumber
S Q(y+n)&=&e\psi^0(y+n)\sum_m G\na\sigma^0(m,y+n)Q(m)\\
\nonumber
&=&e\psi^0(y)\sum_mG\na\sigma(y+n-m)Q(m)
\eeqn
due to the $\Gamma$-periodicity of $\psi^0$. 
Applying  the Fourier transform (\re{F}), we obtain \eqref{S-act}.
\end{proof}
%%%%%%%%%%%%%%%%%%%%%%%%%%%%%%%%%%%%%%%%%%%%%%%%%%%%%%%%%%%%%%%%%%%%%%%%%%%%%%%%%
Furthermore, $\hat S^{\5*}(\theta)$ in (\re{tiAh}) is the corresponding adjoint operator,
and   $\hat T(\theta)$ is the operator matrix expressed by (\re{K3}).
Note that $\hat S(\theta)$, $\hat S^{\5*}(\theta)$ and  $\hat T(\theta)$
are finite dimensional operators.
%%%%%%%%%%%%%%%%%%%%%%%%%%%%%%%

%%%%%%%%%%%%%%%%%%%%%%%%%%%%%%%%%%%%%%%%%%%%%%%%%%%%%%%%%%%%%%%%%%%
\subsection{Generator in the Bloch transform}

\begin{definition}\la{FGLZ}
The Bloch transform of $Y$ is defined as
\be\la{YPi}
 \ti Y(\theta)=[\cF Y](\theta):= \cM(\theta)\hat Y(\theta):
 =(\ti{\Psi}_1(\theta,y),\ti{\Psi}_2(\theta,y),\hat{ Q}(\theta),\hat{ P}(\theta)),
 \qquad~~{\rm a.e.}~~\theta\in\R^3,
\ee
where $\ti{\Psi}_j(\theta,y)=M(\theta)\hat{\Psi}_j:=e^{i\theta y}\hat{\Psi}_j(\theta,y)$
are  $\Ga$-periodic functions in $y\in\R^3$. 
%and $\Ga^*$-quasiperiodic in $\theta\in\R^3$.
\end{definition}
Now the Parseval-Plancherel identities (\re{PP}) read
\be\la{PPt}
 \Vert Y\Vert_\cV^2= |\Pi^*|^{-1}\Vert\ti Y\Vert_{L^2(\Pi^*,\cV(T^3))}^2,
 \qquad
 \Vert Y\Vert_\cX^2 = |\Pi^*|^{-1}\Vert\ti Y\Vert_{L^2(\Pi^*,\cX(T^3))}^2.
\ee
Hence, $\cF:\cX\to  L^2(\Pi^*,\cX(T^3))$ is the isomorphism.
The inversion is given by 
\be\la{FZI}
Y(n)=|\Pi^*|^{-1}\int_{\Pi^*}e^{-in\theta}\cM(-\theta)\ti Y(\theta)d\theta,\qquad n\in\Z^3.
\ee
%%%%%%%%%%%%%%%%%%%%%%%%%%%%%%%%%%%%%%%%%%%%%%%%%%%%%%%%%%%%
Finally, the  above calculations can be summarised as follows:
(\ref{CPFh}) implies that for $Y\in\cD$
\be\la{CPF}
\widetilde{AY}(\theta)=
\ti A(\theta) \ti Y(\theta),\quad~~{\rm a.e.}~~\theta\in  \Pi^*\setminus\Ga^*.
\ee
Here
\be\la{tiA}
\ti A(\theta)\!=\!\cM(\theta)\hat A(\theta)\cM(-\theta)
 \!=\!\!\left(\!\!\!
 \ba{cccl}
0 &\!\!\ti H^0(\theta) & 0  & 0\medskip\\
 -\ti H^0(\theta)-2e^2\psi^0\ti G(\theta)\psi^0&
0 &~\ti S(\theta) & 0\\
 0  & 0&\!0&\!M^{-1}\\
 -2\ti S^{\5*}(\theta)\!\!&  0 &-\hat T(\theta)&\!0\\
\ea\!\!\!\!\right),
\ee
where 
\beqn\la{tiHS}
\ti S(\theta)&:=& M(\theta)\hat S(\theta)=
e\psi^0\ti G(\theta)\na\ti\sigma^0(\theta),\\
\nonumber\\
\ti H^0(\theta)&:=&M(\theta) H^0M(-\theta)=-\fr 12(\na+i\theta)^2-e\Phi^0(x)-\om^0,\la{tiH0}\\
\nonumber\\
\la{tiH1}
\ti G(\theta)&:=&M(\theta) \hat G(\theta) M(-\theta)=(i\na-\theta)^{-2}.
\eeqn
\begin{remark}\la{rb}
The operators $\ti G(\theta):L^2(T^3)\to H^2(T^3)$ are bounded for
$\theta\in \Pi^*\setminus\Ga^*$.
\end{remark}
%%%%%%%%%%%%%%%%%%%%%%%%%%%%%%%%%%%%%%%
%%%%%%%%%%%%%%%%%%%%%%%%%%%%%%%%%%%%%%%
%%%%%%%%%%%%%%%%%%%%%%%%%%%%%%%%%%%%%%%
%%%%%%%%%%%%%%%%%%%%%%%%%%%%%%%%%%%%%%%
%%%%%%%%%%%%%%%%%%%%%%%%%%%%%%%%%%%%%%%
%%%%%%%%%%%%%%%%%%%%%%%%%%%%%%%%%%%%%%%
%%%%%%%%%%%%%%%%%%%%%%%%%%%%%%%%%%%%%%%
%%%%%%%%%%%%%%%%%%%%%%%%%%%%%%%%%%%%%%%

%%%%%%%%%%%%%%%%%%%%%%%%%%%%%%%%%%%%%%%%%%%%%%%%%%%%%%%%%%%%%%%%%%%%%%%%%%%%%%%%%%%%
%%%%%%%%%%%%%%%%%%%%%%%%%%%%%%%%%%%%%%%%%%%%%%%%%%%%%%%%%%%%%%%%%%%%%%%%%%%%%%%%%%%%%

%%%%%%%%%%%%%%%%%%%%%%%%%%%%%%%%%%%%%%%%%%%%%%%%%%%%%%%%%%%%%%%%%%%%%%%%%%%%%%%%%%%%
%%%%%%%%%%%%%%%%%%%%%%%%%%%%%%%%%%%%%%%%%%%%%%%%%%%%%%%%%%%%%%%%%%%%%%%%%%%%%%%%%%%%%

%%%%%%%%%%%%%%%%%%%%%%%%%%%%%%%%%%%%%%%%%%%%%%%%%%%%%%%%%%%%%%%%%%%%%%%%%%%%%%%%%%%%
\begin{lemma}\la{lB} Let the condition (\re{Hpos2}) hold. Then
the operator  $\ti A(\theta)$ admits the representation 
\be\la{Has}
\ti A(\theta)=J\ti B(\theta),~~~~~~~~~~\theta\in \Pi^*\setminus\Ga^*,
\ee
where $\ti B(\theta)$ is the selfadjoint operator (\re{hess2i}) 
in $\cX(T^3)$
with the domain 
\be\la{tiD} 
\ti D:=H^2(T^3)\oplus H^2(T^3)\oplus \C^3\oplus   \C^3.
\ee
\end{lemma}
%%%%%%%%%%%%%%%%%%%%%%%%%%%%%%%%%%%%%%%%%%%%%%%%%%%%%%%%%%%%%%%%%%%%%%%%%%%%%%%%%%%%%
\Pr 
The representation (\re{Has}) follows from  (\re{AJB}) and (\ref{AJB}).
The operator $\ti B(\theta) $
is symmetric on the domain  $\ti D$.
Moreover, operators in (\re{hess2i}) are all bounded, except for $\ti H^0(\theta)$,
which is  selfadjoint in $L^2(T^3)$ with the domain $H^2(T^3)$.
Hence, $\ti B(\theta)$ is also selfadjoint on  the domain $\ti D$.
\bo

%%%%%%%%%%%%%%%%%%%%%%%%%%%%%%%%%%%%%%%%%%%%%%%%%%%%%%%%%%%%%%%%%%%%%%%%%%%%%%%%%%%%%%%%%
%%%%%%%%%%%%%%%%%%%%%%%%%%%%%%%%%%%%%%%%%%%%%%%%%%%%%%%%%%%%%%%%%%%%%%%%%%%%%%%%%%%%%%%%
\setcounter{equation}{0}
\section{The positivity  of energy}
%%%%%%%%%%%%%%%%%%%%%%%%%%%%%%%%%%%%%%%%%%%%%%%%%%%%%%%%%%%%%%%%%%%%%%%%%%%%%%%%%%%%%%%%%%
%%%%%%%%%%%%%%%%%%%%%%%%%%%%%%%%%%%%%%%%%%%%%%%%%%%%%%%%%%%%%%%%%%%%%%%%%%%%%%%%%%%%%%%%%%%%%%%%%%%
Here we prove the  positivity (\re{Hpos2}) for the linearized 
dynamics (\re{JDi}) under conditions (\re{W1}) and (\re{Wai}).
It is easy to construct the corresponding examples of densities
 $\si(x)$.

\begin{example}\la{ex}
 (\re{W1}) holds for 
 $\si(x)\in L^1(\R^3)$ if
\be\la{Wex}
\ti\si(\xi)\ne 0,\qquad ~~\mbox{\rm a.e.}~~\xi\in\R^3.
\ee

\end{example}

\begin{example}\la{ex2}
Let us define  the function $f(x)$ by its Fourier transform
$\ti f(\xi):=\ds\fr{2\sin\ds\fr\xi2}{\xi}e^{-\xi^2}$, and set
\be\la{exW}
\si(x):=eZ f(x_1)f(x_2)f(x_3),\qquad x\in\R^3.
\ee
Then $\si(x)$ is the
smooth function satisfying 
 the Wiener condition
(\re{W1}), as well as (\re{Wai}) and (\re{ro+}), and
\be\la{rd}
 |\si(x)|\le C(a)e^{-a|x|},~~~~x\in\R^3,
\ee
for any $a>0$ by the Paley--Wiener theorem.
\end{example}
The matrix  (\re{W1}) is a continuous function of $\theta\in\Pi^*\setminus\Ga^*$.
Let us denote 
\be\la{Pi+}
\Pi^*_+:=\{\theta\in \Pi^*\setminus\Ga^*: \Si(\theta)>0 \}.
\ee
Then the Wiener condition  (\re{W1}) means that $|\Pi^*_+|=|\Pi^*|$.
In the rest of this paper 
we assume condition (\re{Wai})
and
consider the linearized 
dynamics (\re{JDi}) corresponding  to a real minimizer 
of energy per cell. In Appendix B we show  that the real minimizer
exists and is unique.

\bt\la{tpose}
Let  conditions (\re{L123i}), and  (\re{Wai}) hold.
Then the Wiener  condition  (\re{W1}) is necessary  
and sufficient for
the positivity (\re{Hpos2})
 of the generator
 corresponding to the real minimizer 
of energy per cell.

\et
\Pr
i) First, let us check that the Wiener condition (\re{W1}) is necessary.
Namely, let us consider the inequality   (\re{Hpos2}) for $Y_0=(0,0,Q,P)\in \cV(T^3)$:
(\re{hess2i}) and (\re{Hpos2}) imply that 
\be\la{Hpos0}
 \cE(\theta,Y_0)
=Q\hat T(\theta)Q+PM^{-1}P
\ge \vka(\theta)[|Q|^2+|P|^2],\quad~~\mbox{\rm a.e.}~~\theta\in\Pi^*\setminus\Ga^*.
\ee
Formula (\re{K33}) implies that $\hat T_2=0$ 
by (\re{ppo}). Hence, 
\be\la{TT}
\hat T(\theta)=\hat T_1(\theta)=\Si(\theta),\qquad\theta\in\Pi^*\setminus\Ga^*
\ee
by (\re{K3}).
Therefore, (\re{Hpos0}) becomes
\be\la{Hpos0b}
 \cE(\theta,Y_0)
=
Q\Si(\theta)Q+PM^{-1}P
\ge \vka(\theta)[|Q|^2+|P|^2].
\ee
Hence, the  condition  (\re{W1})
is necessary
for the positivity (\re{Hpos2}).
\medskip\\
ii) It remains to show  that the Wiener condition (\re{W1}) together with (\re{Wai})
is sufficient for the 
positivity 
(\re{Hpos2}).
Let us
 translate the calculations (\re{bb2})--(\re{fg})
into the Fourier--Bloch transform. The operators (\re{fg}) commute with the $\Ga$-translations, 
and therefore
\be\la{fact}
e^2\psi^0\ti G(\theta)\psi^0=\ti f^*(\theta)\ti f(\theta),
\quad\ti S(\theta)=\ti f^*(\theta)\ti g(\theta),
\quad\hat T_1(\theta)= \ti g^*(\theta)\ti g(\theta), 
\ee
where
$\ti f(\theta):=e\sqrt{\ti G(\theta)}\psi^0$ and
$\ti g(\theta)={\sqrt{\ti G(\theta)}}\na\ti\si(\cdot,\theta)$.
Hence,  (\re{hess2i})
 implies 
that
\be\la{bb2t}
\cE(\theta,Y)\!:=\!
 \langle  Y,\ti B(\theta) Y \rangle_{T^3}\!=\!b(\theta, \Psi_1,Q)+
 2\langle  \Psi_2,\ti H^0(\theta) \Psi_2 \rangle_{T^3}
+ PM^{-1} P,\quad Y\!=\!(\Psi_1, \Psi_2, Q, P)\in {\cal V}(T^3),
\ee
where  
\be\la{sq}
b(\theta, \Psi_1,Q)
:=
2\langle  \Psi_1,\ti H^0(\theta) \Psi_1 \rangle_{T^3}+
\langle 2\ti f(\theta)   \Psi_1 + \ti g(\theta) Q, 
~2\ti f(\theta)   \Psi_1+\ti g(\theta) Q \rangle_{T^3}.
\ee
Let us note that $\ti H^0(\theta)=
-\ds\fr12(\na+i\theta)^2$ by (\re{ppo}). Hence,
the eigenvalues of $\ti H^0(\theta)$ equal to $\ds\fr12|2\pi m-\theta|^2$ where 
$m\in \Z^3$. Therefore, 
$\ti H^0(\theta)$ is positive definite:
\be\la{H0d}
\langle\Psi_1, \ti H^0(\theta)\Psi_1 \rangle\ge \fr12 d^2(\theta)\Vert  \Psi_1\Vert_{H^1(T^3)}^2~,
\qquad 
\theta\in\Pi^*\setminus \Ga^*,
\ee
where $d(\theta):=\dist(\theta,\Ga^*)$.
Hence, 
it remains to prove the following proposition.

\bp\la{p}
Under conditions of Theorem \re{tpose} 
\be\la{ub2}
b(\theta,  \Psi_1,Q)
\ge \ve(\theta)[\Vert\Psi_1\Vert_{H^1(T^3)}^2+| Q|^2],\qquad 
\theta\in
\Pi^*_+,
\ee
 where $\ve(\theta)>0$. 

\ep
\Pr 
%Our arguments are parallel to  the proof of 
%Silvester's criterion for $2\times 2$ matrices. 
Let us denote 
 $\al:=\langle  \Psi_1,\ti H^0(\theta) \Psi_1 \rangle_{T^3}$,
 and 
\be\la{abb}
\beta_{11}:=
\langle 2\ti f(\theta)   \Psi_1,
2\ti f(\theta)   \Psi_1\rangle_{T^3},
\quad
\beta_{12}:=
 \langle 2\ti f(\theta)   \Psi_1, \ti g(\theta) Q \rangle_{T^3},
\quad
\beta_{22}:=
\langle \ti g(\theta) Q, \ti g(\theta) Q \rangle_{T^3}.
%=
%\langle Q, \hat T(\theta) Q \rangle_{T^3}.
\ee
Then we can write the quadratic form (\re{sq}) as 
\be\la{sq2}
b=
2\al+\beta,\qquad
\beta:=
\beta_{11}+2\rRe\beta_{12}+\beta_{22}.
\ee
The positivity  (\re{H0d}) implies that 
\be\la{abb3}
\al\ge \de(\theta)
\beta_{11},\qquad \theta\in\Pi^*\setminus\Ga^*,
\ee
where $\de(\theta)>0$.
 Hence, 
\be\la{abb4}
b\ge\al+
(1+ \de(\theta))
\beta_{11}+2\rRe\beta_{12}+\beta_{22},\qquad \theta\in\Pi^*\setminus\Ga^*.
\ee
On the other hand, the Cauchy-Schwarz inequality implies that
\be\la{abb5}
|\beta_{12}|\le \beta_{11}^{1/2}\beta_{22}^{1/2}\le \fr12[\ga\beta_{11}+\fr1\ga\beta_{22}]
\ee
for any $\ga>0$.
Hence, 
\be\la{abb7}
b\ge\al+
(1+ \de(\theta)-\ga)
\beta_{11}+(1-\fr1{\ga})\beta_{22},\qquad \theta\in\Pi^*\setminus\Ga^*.
\ee
Therefore, 
choosing $1<\ga\le 1+ \de(\theta)$,
we obtain (\re{ub2})  from (\re{H0d}) 
since
\be\la{b22}
\beta_{22}=
Q\hat T_1(\theta)Q= \Si(\theta)|Q|^2
\ee
by (\re{fact}) and (\re{TT}).\bo

%%%%%%%%%%%%%%%%%%%%%%%%%%%%%%%%%%%%%%%%%%%%%%%%%%%%%%%%%%%%%%%%%%%%%%%%%%%%
%%%%%%%%%%%%%%%%%%%%%%%%%%%%%%%%%%%%%%%%%%%%%%%%%%%%%%%%%%%%%%%%%%%%%%%%%
\setcounter{equation}{0}
\section{Weak solutions and linear stability}\la{sred}
%%%%%%%%%%%%%%%%%%%%%%%%%%%%%%%%%%%%%%%%%%%%%%%%%%%%%%%%%%%%%%%%%%%%%%%%%
%%%%%%%%%%%%%%%%%%%%%%%%%%%%%%%%%%%%%%%%%%%%%%%%%%%%%%%%%%%%%%%%%%%%%%%%%%

%%%%%%%%%%%%%%%%%%%%%%%%%%%%%%%%%%%%%%%%%%%%%%%%%%%%%%%%%%%%%%%%%%%%%%
Weak solutions are introduced and the linear stability is proved.

\subsection{Weak solutions }
We will consider solutions to (\re{JDi}) in the sense of distributions.
Let us recall that $A^*V\in\cX$ for $V\in\cD$ by Corollary \re{cAA*}.

\bd\la{dws}
$Y(t)\in C(\R,\cX)$ is a weak solution to (\re{JDi}) if
\be\la{ws}
-\int\langle Y(t),\dot\vp(t)V\rangle dt= \int\langle Y(t),\vp(t)A^*V \rangle dt,\qquad \vp\in C_0^\infty(\R),~~V\in\cD.
\ee

\ed
Let us translate this definition into the Fourier--Bloch transform: by  the  Parseval--Plancherel identity
\be\la{wsF}
-\int\Big[\int_{\Pi^*}\langle \ti Y(\theta, t),\dot\vp(t)\ti V(\theta)\rangle_{T^3} d\theta\Big] dt= \int\Big[\int_{\Pi^*}
\langle \ti Y(\theta,t),\vp(t)\ti A^*(\theta)\ti V(\theta) \rangle_{T^3} d\theta\Big] dt
\ee
Respectively, (\re{ws}) is equivalent to the identity
\be\la{wsF2}
\!-\!\int \langle \ti Y(\theta, t),\dot\vp(t)\ti V\rangle_{T^3} dt\!=\! \int
\langle \ti Y(\theta,t),\vp(t)\ti A^*(\theta)\ti V \rangle_{T^3}  dt,
~~ \vp\!\in\! C_0^\infty(\R),~~\ti V\!\in\! \cD(T^3),~~
\mbox{\rm a.e.}~~\theta\in  \Pi^*\setminus \Ga^*,
\ee
where $\cD(T^3):=C^\infty(T^3)\oplus C^\infty(T^3)\oplus \C^3\oplus\C^3$.
In other words,
\be\la{JD}
\dot {\ti Y}(\theta, t)= \ti A(\theta) \ti Y(\theta,t),\qquad \mbox{\rm a.e.}~~\theta\in  
\Pi^*\setminus \Ga^*
\ee
in the sense of vector-valued distributions.

\subsection{Linear stability}
The equation (\re{JD}) is equivalent to
\be\la{CPF2}
 \dot{\ti Y}(\theta,t)=J\ti B(\theta) \ti Y(\theta,t)~,\quad t\in\R,\quad~~\mbox{\rm a.e.}~~\theta\in\Pi^*\setminus \Ga^*.
\ee
We 
reduce it, 
using (\re{Hpos2}),
to an equation with 
a selfadjoint generator by our methods \ci{KK2014a,KK2014b} which is an infinite-dimensional version
of some Gohberg and Krein ideas from the theory of parametric resonance \ci[Chap. VI]{GK}.
We reproduce some details of \ci{KK2014a} for the convenience of the reader.
Namely,  let us denote 
\be\la{tLa}
\ti\Lam(\theta)=\ti B^{1/2}(\theta)>0,\qquad \theta\in\Pi^*_+.
\ee
This is a selfadjoint operator with the domain $\cV(T^3)$,
that follows by the interpolation arguments, 
and  the range $\Ran\5\ti\Lambda(\theta)=\cX(T^3)$.
Its inverse is bounded  in $\cX(T^3)$  by (\re{Hpos2}),  and
\be\la{nVet}
 \Vert \ti\Lam^{-1}(\theta)Z\Vert_{\cV(T^3)}\le \fr 1{\sqrt{ \vka(\theta)}} \Vert Z\Vert_{\cX(T^3)},
 \quad Z\in\cX(T^3),\qquad  \theta\in\Pi^*_+.
\ee
Let us  set
$\ti Z(\theta,t):=\ti\Lam(\theta)\ti Y(\theta,t)$, and now
equation (\re{CPF2}) implies that
\be\la{CPF3}
 \dot {\ti Z}(\theta,t)=-i\ti K(\theta)\ti Z(\theta,t),\quad t\in\R,\quad
  ~~\mbox{\rm a.e.}~~\theta\in\Pi^*_+
\ee
in the sense of vector-valued distributions,
where  $\ti K(\theta)=i\ti\Lam(\theta) J\ti\Lam(\theta)$.

%%%%%%%%%%%%%%%%%%%%%%%%%%%%%%%%%%%%%%%%%%%%%%%%%%%%%%%%%%%%%%
\begin{lemma}\la{lH0} {\rm (Lemma 2.1 of \ci{KK2014a})}
$K(\theta)$ is a selfadjoint operator in $ \cX(T^3)$ with a dense domain  $D(K(\theta))\subset \cV(T^3)$
for every $\theta\in\Pi^*_+$.
\end{lemma}
%%%%%%%%%%%%%%%%%%%%%%%%%%%%%%%%%%%%%%%%%%%%%%%%%%%%%%%%%%%%%%%%%%%
\Pr
The operator $\ti K(\theta)$ is injective. On the other hand, $\Ran\5\ti\Lambda(\theta)=\cX(T^3)$, and
$J:\cX(T^3)\to\cX(T^3)$ is a bounded invertible operator.
Hence, $\Ran\5\ti  K(\theta)=\cX(T^3)$. Consider the inverse operator
\begin{equation}\la{G}
 \ti R(\theta):=\ti K^{-1}(\theta)=i\ti\Lam^{-1}(\theta) J\ti\Lam^{-1}(\theta).
\end{equation}
It is selfadjoint since $D(\ti R(\theta))=\Ran\5 K(\theta)=\cX(T^3)$ and $\ti R(\theta)$ is bounded and symmetric.
Finally, $\ti R(\theta)$ is injective, and hence, $\ti K(\theta)=\ti R^{-1}(\theta)$ is a densely defined selfadjoint operator
by Theorem 13.11 (b) of \cite{Rudin}:
$$
\ti K^*(\theta)=\ti K(\theta)~, \quad D(\ti K(\theta))=\Ran\5 \ti R(\theta)\subset \
\Ran \5\ti\Lam^{-1}(\theta)\subset \cV(T^3)
$$
%%%%%%%%%%%%%%%%%%%%%%%%%%%%%%%%%%%%%%%%%%%%%%%%%%%%%%%%%%%%%%%%%%%%%%%%%%%5
by (\re{nVet}).\bo
\medskip\\
This lemma implies that each weak solution to (\re{CPF3}) is given by
\be\la{ZH0}
 \ti Z(\theta,t)=e^{-i \ti K(\theta) t}\ti Z(\theta,0) \in   C_b(\R,\cX(T^3)),\qquad ~~\mbox{\rm a.e.}~~\theta\in\Pi^*_+
\ee
for $\ti Z(\theta,0)\in\cX(T^3)$.
Hence, we obtain 
the well posedness of the Cauchy problem 
for  equation (\re{CPF2}).

%%%%%%%%%%%%%%%%%%%%%%%%%%%%%%%%%%%%%%%%%%%%%%%%%%%%%%%%%%%%%%%%%%%%%%

\bt\la{tdt}
Let  all conditions of Theorem \re{tpose} hold and $\theta\in\Pi^*_+$. Then
for every initial state $\ti Y(\theta,0)\in \cV(T^3)$ there exists a unique
weak solution $\ti Y(\theta,t)\in C_b(\R,\cV(T^3))$ to equation (\re{CPF2}), and
\be\la{econ}
 \langle\ti \Lam(\theta) \ti Y(\theta,t),\ti \Lam(\theta)\ti Y(\theta,t)\rangle_{T^3}=\const,
\qquad t\in\R.
\ee
\et
%%%%%%%%%%%%%%%%%%%%%%%%%%%%%%%%%%%%%%%%%%%%%%%%%%%%%%%%%%%%%%%%%%%%%%%%%%%
\Pr 
$\ti Z(\theta,0):=\ti\Lam(\theta)\ti Y(\theta,0)\in \cX(T^3)$
since 
$Y(\theta,0)\in \cV(T^3)$. Hence,
(\re{ZH0}) 
and (\re{nVet}) 
imply that
\be\la{YH0}
  \ti Y(\theta,t)=\ti\Lam^{-1}(\theta) e^{-i K(\theta) t} \ti Z(\theta,0) \in   C_b(\R,\cV(T^3)).
\ee
Finally, (\re{econ}) holds since  $e^{-i K(\theta) t}$  is  the unitary group 
in $\cX(T^3)$,
and hence
$$
\langle\ti \Lam(\theta)\ti  Y(\theta,t),\ti \Lam(\theta)\ti Y(\theta,t)\rangle_{T^3}=
\langle\ti Z(\theta,t),\ti Z(\theta,t)\rangle_{T^3}=\const,\qquad t\in\R.
\qquad \qquad\qquad \qquad
\bo
$$

Now we apply this theory to equation  (\re{JDi}).
Let us note that $\ti \Lam(\theta)\ti Y(\theta)\in L^2(\Pi^*_+,\cX(T^3))$
for 
 $Y\in\cD$, see Definition \re{dD}.

\bd\la{dW}
The Hilbert space $\cW$ is the completion of $\cD$ in the norm
\be\la{nW}
\Vert Y \Vert_\cW:=\Vert \ti\Lam(\theta)\ti Y(\theta) \Vert_{L^2(\Pi^*_+,\cX(T^3))}
\ee
\ed
Formally, $\Vert Y\Vert_\cW=\langle Y,BY\rangle^{1/2}$.
The Fourier-Bloch transform (\re{YPi}) extends to the isomorphism
\be\la{FBW}
\cF:\cW\to \ti\cW:=\{\ti Y(\cdot)\in L^2_{\rm loc}(\Pi^*_+,\cX(T^3)): 
\Vert \ti\Lam(\theta)\ti Y(\theta) \Vert_{L^2(\Pi^*_+,\cX(T^3))}<\infty\}.
\ee
Finally, let us extend definition of  weak solutions to  $Y(t)\in C_b(\R,\cW)$
by the identity  (\re{CPF2}) in the sense of vector-valued distributions (\re{wsF2}).
Then  Theorem  \re{tdt}  implies the following corollary.
\bc\la{cd}
Let all conditions of Theorem \re{tpose} hold. Then
for every initial state $Y(0)\in \cW$ there exists a unique
weak solution $Y(t)\in C_b(\R,\cW)$ to equation (\re{JDi}), and the energy norm is conserved{\rm :}
\be\la{econ2}
 \Vert  Y(t)\Vert_\cW=\const,\qquad t\in\R.
\ee
The solution is given by the formula (\re{YH0}):
\be\la{tZ0}
 Y(t)=\cF^{-1}\ti\Lam^{-1}(\theta) e^{-i K(\theta) t} \ti Z(\theta,0) \in   C_b(\R,\cW(T^3)).
\ee
\ec
This means that the linearized dynamics 
(\re{JDi})
is stable: global 
solutions exist for all initial states of finite energy, and the norm is constant in time.

%%%%%%%%%%%%%%%%%%%%%%%%%%%%%%%%%%%%%%%%%%%%%%%%%%%%%%%%%%%%%%%%%%%%%%%%%%%

%%%%%%%%%%%%%%%%%%%%%%%%%%%%%%%%%%%%%%%%%%%%%%%%%%%%%%%%%%%%%%%%%%%%%%%%%%%%%%%%%%%%%%
%%%%%%%%%%%%%%%%%%%%%%%%%%%%%%%%%%%%%%%%%%%%%%%%%%%%%%%%%%%%%%%%%%%%%%%%%%%%%%%%%%%%%%%

\setcounter{equation}{0}
\section{Examples of negative energy}

We show that the positivity (\re{Hpos2}) can fail
if the condition (\re{Wai})
breaks down even when  the Wiener condition  (\re{W1}) holds.
Namely,  for $Y_0=(0,0,Q,0)\in \cV(T^3)$ we have
\be\la{Hpos02}
 \cE(\theta,Y_0)
=Q\hat T(\theta)Q
\ee
by (\re{Hpos0}).
\bl\la{lne}
There exist functions $\mu(x)$ 
such that 
the positivity (\re{Hpos2})  fails 
for $\si(x)$ from (\re{fam}) 
with small $e>0$ while 
(\re{L123i}) and the Wiener condition  (\re{W1}) hold. 
\el
\Pr
It suffices to construct an example of $\si(x)$ which provides 
\be\la{Tt0}
Q\hat T(\theta_0)Q<0
\ee
for some
$\theta_0\in\Pi^*\setminus\Ga^*$ and 
$Q\in\C^3$.
The representation (\re{K3}) can be written as 
\beqn\la{K322}
\hat T_1(\theta)
=
e^2 \sum_m
\Big[\fr{\xi\otimes\xi}{|\xi|^2}|\ti\mu(\xi)|^2\Big]_{\xi=2\pi m-\theta}~,
\quad\theta\in\Pi^*\setminus\Ga^* 
\eeqn
Similarly, (\re{K33}) can be written in the Fourier representation as 
\be \la{K332}
\hat T_2
=-e^2\fr1{(2\pi)^3}\langle \ti\nu^0_e(\xi) \fr{\xi\otimes\xi}{|\xi|^2},
\ti\mu(\xi)\rangle  
\ee
 with
$\nu^0_e:=
\mu_{\rm per}(x)-|\psi_e^0(x)|^2$  according to  (\re{L5}). 
The asymptotics (\re{p0}) of the ground state $\psi^0_e(x)$ implies
\be\la{W4}
\ti\nu^0_e(\xi)= \ti\mu_{\rm per}(\xi)-|\ga_e|^2(2\pi)^3\de(\xi)-\ti s(\xi)=
\ti\mu_{\rm per}(\xi)-Z(2\pi)^3\de(\xi)-\ti s(\xi),
\ee
since $|\ga_e|^2=Z$ by  (\re{p0}).
Here $s(x)=\ga_e\ov\chi_e(x)+\ov \ga_e\chi_e(x)+|\chi_e(x)|^2$, and
\be\la{W4e}
\Vert s\Vert_{L^2(T^3)}\le C_1 e^2
\ee
by (\re{p0}). 
Further, (\re{muj}) gives 
\be\la{W5}
\ti\mu_{\rm per}(\xi)=\sum_n\ti\mu(\xi)e^{in\xi}=
\ti\mu(\xi)(2\pi)^3\sum_m\de(\xi-2\pi m)
%=\ti\mu(0)(2\pi)^3\de(\xi)
\ee
by the Poisson summation formula \ci{Her1}.
% and (\re{Wa}). 
Substituting (\re{W5}) into  (\re{W4}) we get
\be\la{W42}
 \ti\nu^0_e(\xi)
% (2\pi)^3[\ti\mu(0)-Z]\de(\xi)-\ti s(\xi)
=\ti\mu(\xi)(2\pi)^3\sum_{m\ne 0} \de(\xi-2\pi m)-\ti s(\xi)
\ee
by (\re{ro+}) and (\re{fam}). 
Substituting this expression  into  (\re{K33}) we obtain
\be\la{K34}
\hat T_2=
-e^2\langle \ti \mu(\xi) 
\sum_{m\ne 0} \de(\xi-2\pi m)
\fr{\xi\otimes\xi}{|\xi|^2},
\ti\mu(\xi)\rangle
+
\fr{e^2}{(2\pi)^3}\langle \ti s(\xi) \fr{\xi\otimes\xi}{|\xi|^2},
\ti\mu(\xi)\rangle.
\ee
At last, 
$s(x)$ is a $\Ga$-periodic function and 
$$\int_{T^3} s(x)dx=\int_{T^3} \nu^0_e(x)dx=0
$$
by (\re{LPS20e}). Hence, 
\be\la{rp}
\ti s(\xi)=\sum_{m\ne 0} \check s(m)\de(\xi-2\pi m),\qquad \sum_m|\check s(m)|^2=\cO(e^4),\qquad e\to 0
\ee
by (\re{W4e}). 
Therefore,
\be\la{K342}
\hat T_2=
-e^2\sum_{m\ne 0 } \Big[\fr{\xi\otimes\xi}{|\xi|^2}|\ti \mu(\xi)|^2\Big]_{\xi=2\pi m}+\cO(e^4),\qquad e\to 0.
\ee
Hence, there exists a $Q\in\C^3$ such that
\be\la{T2n}
Q\hat T_2Q<0
 \ee
 for small $e>0$ if the condition (\re{Wa})
breaks down.
For example, we can take $Q=2\pi m$ with $m\in \Z^3\setminus 0$ if 
$\ti\mu(2\pi m)\ne 0$.
 Finally,  
for any $\theta_0\not\in\Ga^*$ 
we can reduce $|\hat \mu(\theta)|$
in all points  $\theta\in \theta_0+\Ga^*$ 
keeping it in the points of $\Ga^*$
to have 
\be\la{T2neg}
 Q\hat T(\theta_0)Q=Q\hat T_1(\theta_0)Q+Q\hat T_2Q<0.
 \ee
At the same time, we can
keep (\re{L123i}) and the Wiener condition (\re{W1}) to hold.\bo

\br\la{rT2}
The operator $T_2$ corresponds to the last term in the last line
of  (\re{LPS1Li}). This term describes the "virtual repulsion" of the ion located at $n+q^0$  from the same ion 
deflected  to the point $n+q^0+Q(n,t)$. This means that  the negative energy contribution 
is provided by the electrostatic instability ("Earnshaw Theorem" {\rm \ci{Stratton}}).

\er

%%%%%%%%%%%%%%%%%%%%

%%%%%%%%%%%%%%%%%%%%%%%%%%%%%%%%%%%%%%%%%%%%%%%%%%%%%%%%%%%%%%%%%%%%%%%%%%%%%%%%%
%%%%%%%%%%%%%%%%%%%%%%%%%%%%%%%%%%%%%%%%%%%%%%%%%%%%%%%%%%%%%%%%%%%%%%%%%%%%%%%%%
\appendix

\protect\renewcommand{\thesection}{\Alph{section}} 
\protect\renewcommand{\theequation}{\thesection.\arabic{equation}} 
\protect\renewcommand{\thesubsection}{\thesection.\arabic{subsection}} 
\protect\renewcommand{\thetheorem}{\Alph{section}.\arabic{theorem}}

\setcounter{equation}{0}
\section{Formal linearization at the ground state}
%%%%%%%%%%%%%%%%%%%%%%%%%%%%%%%%%%%%%%%%%%%%%%%%%%%%%%%%%%%%%%%%%%%%%%%%%%%%%%%
%%%%%%%%%%%%%%%%%%%%%%%%%%%%%%%%%%%%%%%%%%%%%%%%%%%%%%%%%%%%%%%%%%%%%%%%%%%%%%%
Let us substitute 
\[
 \psi(x,t)=[\psi^0(x)+\Psi(x,t)]e^{-i\om^0 t},~~~~ q(n,t)=q^0+Q(n,t)
\]
into the nonlinear equations (\re{LPS1}), (\re{LPS3}) with $\Phi(x,t)=G\rho(x,t)$. 
First, \eqref{LPS2} implies that
\[
\rho(x,t)=\sum _n\si(x-n-q^0-Q (n,t))-e|\psi^0(x)+\Psi(x,t)|^2
\]
and the Taylor expansion {\it formally} gives
\beqn\la{Ta}
\!\!\!\!  \rho(x,t)\!\!&\!\!=\!\!&\!\!\sum_n\Big[\si(x-n-q^0)-\na\si(x-n-q^0) Q(n,t)+
 \fr12\na\na\si(x-n-q^0)Q(n,t)\otimes Q(n,t)+... \Big]
 \nonumber\\\nonumber\\
\!\!\!\! \!\! &\!\!-\!\!&\!\!e\Big[|\psi^0(x)|^2+2\rRe(\psi^0(x)\ov{\Psi}(x,t))+|\Psi(x,t)|^2\Big]
 =\rho^0(x)+\rho_1(x,t)+\rho_2(x,t)+...
\eeqn
Here 
$\rho^0(x):=\si^0(x)-e|\psi^0(x)|^2$ and 
$\rho_k$ are   polynomials in
$\Psi(x,t)$ and $Q(t)$ of  degree $k$. In particular, 
$\rho_1(x,t)$ is given by  (\re{ro1i}).
As a result, we obtain the system (\re{LPS1Li}) in the linear approximation.

%%%%%%%%%%%%%%%%%%%%%%%%%%%%%%%%%%%%%%%%%%%%%%%%%%%%%%%%%%%%%%%%%%%%%%%%%%%%%%%%%%%%%%

\section{Ground states with minimal energy per cell}

Let us consider any  ion density $\si(x)\in L^2(\R^3)$ satisfying (\re{Wai}):
 \be\la{Wa}
\ti\si(2\pi m)=0,\quad m\in\Z^3\setminus 0.
\ee
Let us note that 
\be\la{Z0}
\ti\si(0)=\int \si(x)dx=eZ>0
\ee
by (\re{ro+}).
 Then $\si_{\rm per}(x):=\sum_n\si(x-n)\equiv eZ$ since
 \be\la{PS}
 \check \si_{\rm per}(m)=\int_{T^3} e^{i2\pi m x}\si_{\rm per}(x)dx=
 \int_{R^3} e^{i2\pi m x}\si(x)dx= 0,\qquad m\in\Z^3\setminus 0
 \ee
 by (\re{Wa}). Therefore, 
 the functions
 \be\la{ppo}
\psi^0(x)\equiv \sqrt{Z},\qquad
\Phi^0(x)\equiv 0,\qquad\om^0=0
\ee
give a solution to (\re{LPS10})--(\re{LPS30}) with zero energy per cell (\re{U}).
On the other hand, the energy (\re{U}) is nonnegative.
Hence,
 the set of all minimizers of energy per cell
 consists of 
$\psi^0(x)\equiv e^{i\phi}\sqrt{Z}$, with $\phi\in [0,2\pi].$

%%%%%%%%%%%%%%%%%%%%%%%%%%%%%%%%%%%%%%%%%%%%%%%%%%%%%%%%%%%%%%%%%%%%%%%%%%%%%%%%%%%%%%

%%%%%%%%%%%%%%%%%%%%%%%%%%%%%%%%%%%%%%%%%%%%%%%%%%%%%%%%%%%%%%%%
%%%%%%%%%%                  bibliography
%%%%%%%%%%%%%%%%%%%%%%%%%%%%%%%%%%%%%%%%%%%%%%%%%%%%%%

\end{document}